\documentclass[a4paper,11pt]{article}
\usepackage[T1]{fontenc}
\usepackage{fullpage}
\usepackage{graphicx,amsmath,amssymb,mathrsfs}
\usepackage{latexsym}
\usepackage{hyperref}
\usepackage{cite}
\usepackage[linesnumbered, vlined, ruled]{algorithm2e}
\usepackage{amsthm}
\usepackage[nodayofweek]{datetime}
\usepackage{enumerate}
\usepackage{multirow}
\usepackage[mathlines]{lineno}
\usepackage{color}
\usepackage{xspace}

\allowdisplaybreaks
\newtheorem{theorem}{Theorem}
\newtheorem{lemma}[theorem]{Lemma}

\newcommand{\radrestricted}[2]{\ensuremath{r^*_{#1 #2}}}
\newcommand{\maxrad}[2]{\max\{\radbd{#1}{#2}, \radbd{#2}{#1}\}}
\newcommand{\radbd}[2]{\ensuremath{\gamma({#1}, {#2})}}
\newcommand{\dintersection}{\ensuremath{\mathcal{I}}}
\newcommand{\boundaryarc}{\ensuremath{\mathcal{A}}}
\newcommand{\centerl}{\ensuremath{c_1}}
\newcommand{\centerra}{\ensuremath{c_2^{\mathrm{c}}}}
\newcommand{\centerrb}{\ensuremath{c_2^{\mathrm{cc}}}}
\newcommand{\CHPQ}{\ensuremath{\mathcal{C}_Q}}
\newcommand{\subpolygon}[2]{\ensuremath{\mathcal{P}_Q{(#1, #2)}}}
\newcommand{\subchain}[2]{\ensuremath{\CHPQ{(#1, #2)}}}
\newcommand{\subchainleft}[2]{\ensuremath{\CHPQ({#2+1},{#1})}}
\newcommand{\subchainright}[2]{\ensuremath{\CHPQ({#1+1},{#2})}}
\newcommand{\spm}[1]{\mathsf{SPM}(#1)}
\newcommand{\fvd}[1]{\mathsf{FVD}(#1)}
\newcommand{\gvd}[1]{\mathsf{VD}(#1)}

\newcommand{\bd}{\ensuremath{\partial}}
\newcommand{\upd}{\textsc{Update}\xspace}
\newcommand{\dcs}{\textsc{Decision}\xspace}
\newcommand{\out}{\textsc{Out}\xspace}
\newcommand{\iin}{\textsc{In}\xspace}
\newcommand{\defpoint}[1]{\ensuremath{\mathit{def}(#1)}}
\newcommand{\inout}[1]{\ensuremath{\mathit{io}(#1)}}
\newcommand{\vcw}[1]{\ensuremath{v_{cw}(#1)}}
\newcommand{\vccw}[1]{\ensuremath{v_{ccw}(#1)}}

\makeatletter \newbox\ProofSym \setbox\ProofSym=\hbox{%
  \unitlength=0.18ex%
  \begin{picture}(10,10) \put(0,0){\framebox(9,9){}}
    \put(0,3){\framebox(6,6){}}
  \end{picture}}
\renewenvironment{proof}[1][Proof.]{\O@proof{#1}}{\O@endproof}
\def\O@proof#1{\trivlist
  \@topsep\z@\@topsepadd\smallskipamount%
  \@ifstar{\item[]}{\item[\hskip\labelsep\it #1 ]}}
\def\O@endproof{\hfill\copy\ProofSym\endtrivlist}
\makeatother

\def\denseitems{
    \itemsep1pt plus1pt minus1pt
    \parsep0pt plus0pt
    \parskip0pt\topsep0pt}

\graphicspath{{figures/}}

\begin{document}


\title{Computing a Geodesic Two-Center of Points in a Simple Polygon%
\thanks{%
  Work by Oh and Ahn was supported by the MSIT(Ministry of Science and ICT), Korea, under the SW Starlab support program(IITP-2017-0-00905) supervised by the IITP(Institute for Information \& communications Technology Promotion.)
  Work by S.W.Bae was supported by Basic Science Research Program
through the National Research Foundation of Korea (NRF) funded
by the Ministry of Education (2015R1D1A1A01057220).
}
}

\author{
Eunjin Oh\thanks{Department of Computer Science and Engineering,
POSTECH, Pohang, South Korea.
Email: {\tt{\{jin9082,heekap\}@postech.ac.kr}}} \and
Sang Won Bae\thanks{Department of Computer Science,
Kyonggi University, Suwon, South Korea.
Email: {\tt{swbae@kgu.ac.kr}}} \and
Hee-Kap Ahn\footnotemark[2]
}
\date{}
\maketitle
\begin{abstract}
  Given a simple polygon $P$ and a set $Q$ of points contained in $P$,
  we consider the geodesic $k$-center problem where we want to find
  $k$ points, called \emph{centers}, in $P$ to minimize the maximum
  geodesic distance of any point of $Q$ to its closest center.
  In this paper, we focus on the case for $k=2$ and present the first
   exact algorithm that efficiently computes an optimal $2$-center of $Q$
  with respect to the geodesic distance in $P$.
\end{abstract}

\section{Introduction}
Computing the centers of a point set in a metric space is a
fundamental algorithmic problem in computational geometry,
which has been extensively studied with numerous applications
in science and engineering.
This family of problems is also known as the \emph{facility location problem}
in operations research that asks an optimal placement of facilities to minimize
transportation costs.
A historical example  is the \emph{Weber problem}
in which one wants to place one facility
to minimize the (weighted) sum of distances from the facility to
input points.
In cluster analysis, the objective is to group
input points in such a way that the points in the same group
are relatively closer to each other than to those in other groups.
A natural solution finds a few number of centers and assign the points to the
nearest center, which relates to the well known \emph{$k$-center problem}.

The $k$-center problem is formally defined as follows:
for a set $Q$ of $m$ points,
find a set $C$ of $k$ points that minimizes
$\max_{q\in Q} \{\min_{c\in C} d(q,c)\}$,
where $d(x,y)$ denotes the distance between $x$ and $y$.
The $k$-center problem has been investigated for point sets in two-, three-,
or higher dimensional Euclidean spaces.
For the special case where $k=1$, the problem is equivalent to finding the
smallest ball containing all points in $Q$. It can be solved in $O(m)$
time for any
fixed dimension~\cite{Euclidean-1-center,higher1-center,plane1-center}.
The case of $k=2$ can be solved in $O(m\log^2 m \log^2 \log m)$ time~\cite{Euclidean-2-center}
in $\mathbb{R}^2$.
If $k>2$ is part of input, it is NP-hard to approximate
the Euclidean $k$-center within approximation factor $1.822$~\cite{FederGreene88},
while an $m^{O(\sqrt{k})}$-time exact algorithm is known
for points in $\mathbb{R}^2$~\cite{Euclidean-k-center}.

There are several variants of the $k$-center problem.  One variant is
the problem for finding $k$ centers in the presence of obstacles. More
specifically, the problem takes a set of pairwise disjoint simple
polygons (obstacles) with a total of $n$ edges in addition to a set
$S$ of $m$ points as inputs.  It aims to find $k$ smallest congruent
disks whose union contains $S$ and whose centers do not lie on the
interior of the obstacles.  Here, the obstacles do not affect the
distance between two points.  For $k=2$, Halperin et
al.~\cite{2-centerobstacle} gave an expected
$O(n\log^2(mn)+mn\log^2m\log(mn))$-time algorithm for this problem.

In this paper, we consider another variant of the $k$-center problem
in which the set $Q$ of $m$ points are given in a simple $n$-gon $P$
and the centers are constrained to lie in $P$.
Here the boundary of the polygon $P$ is assumed to act as
an obstacle and the distance between any two points in $P$ is thus
measured by the length of the geodesic (shortest) path connecting
them in $P$ in contrast to~\cite{2-centerobstacle}.
We call this constrained version \emph{the geodesic
$k$-center problem} and its solution
\emph{a geodesic $k$-center}
of $Q$ with respect to $P$.

This problem has been investigated for the simplest case $k=1$.
The geodesic one-center  of $Q$ with respect to $P$
is proven to coincide with the geodesic one-center of the
geodesic convex hull of $Q$ with respect to $P$~\cite{FVD},
which is the smallest subset $C\subseteq P$ containing $Q$ such that
for any two points $p, q\in C$, the geodesic path between $p$ and $q$
is also contained in $C$.
Thus, the geodesic one-center can be computed by first computing
the geodesic convex hull of $Q$ in $O((m+n) \log (m+n))$ time~\cite{Toussaint89}
and second finding its geodesic one-center.
The geodesic convex hull of $Q$ forms a (weakly) simple polygon with $O(m+n)$ vertices.

Asano and Toussaint~\cite{1-center-first} studied the problem of
finding the geodesic one-center of a (weakly) simple polygon
and presented an $O(n^4\log n)$-time algorithm,
where $n$ denotes the number of vertices of the input polygon.
It was improved to $O(n\log n)$ in~\cite{1-center1989} and finally improved
again to $O(n)$ in~\cite{1-center}.
Consequently, the geodesic one-center of $Q$ with respect to $P$
can be computed in $O((m+n) \log (m+n))$ time.

However, even for $k=2$, finding a geodesic $k$-center of $Q$ with respect to
$P$ is not equivalent to finding a geodesic $k$-center of a (weakly)
simple polygon, which was addressed
in~\cite{2-centersimple}.
One can easily construct an example of $P$ and $Q$ in which the two
geodesic disks corresponding to a geodesic $2$-center of $Q$
do not contain the geodesic convex hull of $Q$.
See \figurename~\ref{fig:outside}.

In this paper, we consider the geodesic $2$-center problem and present
an algorithm to compute a geodesic $2$-center,
that is, a pair $(c_1,c_2)$ of points in $P$ such that $\max_{q \in Q} \{\min\{d(q,c_1),
d(q,c_2)\}\}$ is minimized,
where $d(x,y)$ denote the length of the shortest path connecting $x$ and
$y$ in $P$.
Our algorithm takes $O(m(m+n)\log^3(m+n) \log m)$ time
using $O(m+n)$ space.
If $n$ and $m$ are asymptotically equal,
then our algorithm takes $O(n^2\log^4 n)$ using $O(n)$ space.

\section{Preliminaries}
Let $P$ be a simple polygon with $n$ vertices.  The \emph{geodesic
  path} between $x$ and $y$ contained in $P$, denoted by $\pi(x,y)$,
is the unique shortest path between $x$ and $y$ inside $P$.  We often
consider $\pi(x,y)$ directed from $x$ to $y$.  The length of
$\pi(x,y)$ is called the \emph{geodesic distance} between $x$ and $y$,
and we denote it by $d(x,y)$.

A subset $A$ of $P$ is \emph{geodesically convex} 
if it holds that $\pi(x,y) \subseteq A$ for any $x, y \in A$.
For a set $Q$ of $m$ points contained in $P$, the
common intersection of all the geodesically convex subsets of $P$ that
contain $Q$ is also geodesically convex and it is called the
\emph{geodesic convex hull} of $Q$.  It is already known that the
geodesic convex hull of any set of $m$ points in $P$ is a weakly simple polygon and can be computed in
$O(n+m\log(n+m))$ time~\cite{shortest-path}.

Note that once the geodesic convex hull of $Q$ 
is computed, our algorithm regards the geodesic convex hull as a new
polygon and never consider the parts of $P$ lying outside of the geodesic
convex hull. 
We simply use $\CHPQ$ to denote the geodesic convex hull of $Q$.
Each point $q\in Q$ lying on the boundary of $\CHPQ$ is called \emph{extreme}.

For a set $A$, we use $\bd A$ to denote the boundary of $A$.  Since
the boundary of $\CHPQ$ is not necessarily simple, the clockwise order
of $\bd \CHPQ$ is not defined naturally in contrast to a simple curve.
Aronov et al.~\cite{FVD} presented a way to label the extreme points of
$Q$ with $v_1, \ldots, v_k$ such that the circuit
$\pi(v_1,v_2)\pi(v_2,v_3)\cdots\pi(v_k,v_1)$ is a closed walk of the
boundary of $\CHPQ$ visiting every point of $\bd \CHPQ$ at most twice.
We use this labeling of extreme points for our problem.
The circuit $\pi(v_1,v_2)\pi(v_2,v_3)\cdots\pi(v_k,v_1)$ is called the
\emph{clockwise traversal} of $\bd\CHPQ$ from $v_1$.
The \emph{clockwise order} follows from the clockwise
traversal along $\bd\CHPQ$.

Let $\subchain{v}{w}$ denote the portion of $\bd \CHPQ$ from $v$ to
$w$ in clockwise order (including $v$ and $w$) for $v, w \in \bd
\CHPQ$. For any two extreme points $v_i$ and $v_j$, we use $\subchain{i}{j}$
to denote the chain $\subchain{v_i}{v_j}$ for simplicity.
The subpolygon bounded by $\pi(w,v)$ and $\subchain{v}{w}$ is
denoted by $\subpolygon{v}{w}$.
Clearly, $\subpolygon{v}{w}$ is a weakly simple polygon.

The \emph{geodesic disk} centered at $c \in P$ with radius $r \in
\mathbb{R}$, denoted by $D_r(c)$, is the set of points whose geodesic
distance from $c$ is at most $r$. The boundary of $D_r(c)$ consists
of circular arcs and line segments.
Each circular arc along the boundary of $D_r(c)$ is called a \emph{boundary arc}
of disk $D_r(c)$.
Note that every point lying on a boundary arc of $D_r(c)$ is at distance $r$ from $c$.
A set of geodesic disks with the same radius satisfies the \emph{pseudo-disk
property}.
An extended form of the pseudo-disk property of geodesic disks
can be stated as follows.

\begin{lemma}[Oh et al.~{\cite[Lemma 8]{2-centersimple}}]
  \label{lem:pseudo}
  Let $\mathcal {D}=\{D_1, \ldots, D_k \}$ be a set of geodesic disks
  with the same radius and let $I$ be the common intersection of all
  disks in $\mathcal {D}$. Let $S=\langle s_1, \ldots, s_t\rangle$ be the cyclic
  sequence of the boundary arcs of geodesic disks appearing on $\bd I$ along
  its boundary in clockwise order.  For any $i \in \{1, \ldots, k\}$,
  the arcs in $\bd I \cap \bd D_i$ are consecutive in $S$.
\end{lemma}

For a set $A\subseteq P$, we define the \emph{(geodesic) radius} of $A$
to be $\inf_{c\in P}\sup_{p\in A}d(p,c)$,
that is, the smallest possible radius of a geodesic disk that contains $A$.

We make the \emph{general position assumption} on input points $Q$ with respect to polygon $P$ that
no two distinct points in $Q$ are equidistant from a vertex of $P$.
We make use of the assumption when computing the geodesic Voronoi diagrams of $Q$ with respect to $P$, which are counterparts of the standard Voronoi diagrams with respect to the geodesic distance.
Moreover, this is the only place where we use the general position assumption.
If there is a vertex of $P$ equidistant from two distinct points $v, v'$ of $Q$,
there is a two-dimensional region consisting of the points equidistant from $v$ and $v'$,
which we want to avoid.
Indeed, the general position assumption can be achieved by considering the points
in such a two-dimensional region to be closer to $v$ than to $v'$.
Then every cell of the geodesic nearest-point (and farthest-point) Voronoi diagram is associated with only one site $t$.

\section{Bipartition by Two Centers}

We first compute the geodesic convex hull $\CHPQ$ of the point set $Q$.
Let $Q_B$ be the set of extreme points in $Q$ and let $Q_I := Q\setminus Q_B$.
Note that each $q\in Q_B$ lies on $\bd \CHPQ$ while each $q' \in Q_I$ lies in the interior of $\CHPQ$.
The points of $Q_B$ are readily sorted along the
boundary of $\CHPQ$, being labeled by $v_1,\ldots,v_k$ following
the notion of Aronov~et~al~\cite{FVD}.

 \begin{figure}[ht]
   \begin{center}
     \includegraphics[width=0.5\textwidth]{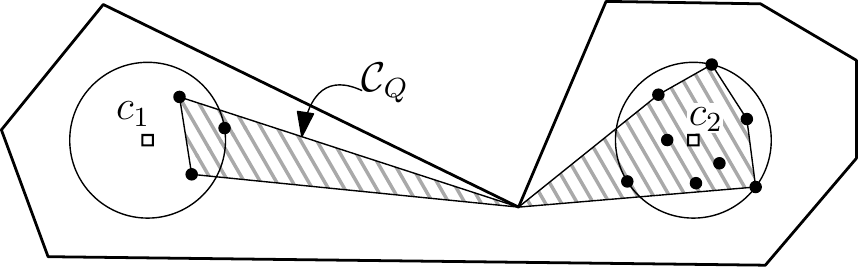}
     \caption{\small The dashed region is the geodesic convex hull $\CHPQ$.
     	The center $c_1$ lies outside of $\CHPQ$, while
     	the center $c_2$ lies inside $\CHPQ$.
       \label{fig:outside}}
     \end{center}
   \end{figure}

Note that it is possible that an optimal two-center has one of its two
centers lying outside of $\CHPQ$ (See \figurename~\ref{fig:outside}).
However,
there always exists an optimal two-center of $Q$ with respect to $P$
such that both the centers are contained in $\CHPQ$ as stated in the
following lemma.
Thus we may search only $\CHPQ$ to find an optimal two-center of $Q$ with
respect to $P$.

\begin{lemma} \label{lem:chpq}
  There is an optimal two-center $(c_1, c_2)$ of $Q$ with respect to $P$ such that
  both $c_1$ and $c_2$ are contained in the geodesic convex hull $\CHPQ$ of $Q$.
\end{lemma}
 \begin{proof}
  Let $(c_1, c_2)$ be an optimal two-center of $Q$ with respect to $P$
  and let $r$ be the smallest radius satisfying $D_r(c_1) \cup D_r(c_2)$ contains $Q$.
  By Corollary 2.3.5 of \cite{FVD},
  the center $c_1'$ of the smallest-radius geodesic disk containing $D_r(c_1)\cap Q$
  lies in the geodesic convex hull of $D_r(c_1)\cap Q$.
  Thus, by definition of the geodesic convexity, $c_1' \subset \CHPQ$
  and $D_r(c_1)\cap Q \subset D_r(c_1')$.
  Similarly, the center $c_2'$ of the smallest-radius geodesic disk containing
  $D_r(c_2) \cap Q$ lies in $\CHPQ$, and $D_r(c_2)\cap Q \subset D_r(c_2')$.
  Thus, $(c_1',c_2')$ is also an optimal two-center of $Q$ with respect to $P$.
  Moreover, both the centers are contained in $\CHPQ$.
\end{proof}

Let $c_1$ and $c_2$ be two points in $\CHPQ$ such that $D_r(c_1) \cup D_r(c_2)$ contains $Q$.
By Lemma~\ref{lem:pseudo}, the boundaries of the two geodesic disks cross
each other at most twice.
If every extreme point on $\subchainleft{i}{j}$ 
is contained in $D_r(c_1)$,
so is the whole chain $\subchainleft{i}{j}$
since $D_r(c_1)$ is geodesically convex and therefore $\pi(v,v')\subset D_r(c_1)$ for
any two extreme points $v$ and $v'$ on $\subchainleft{i}{j}$.
 Thus there exists a pair $(i,j)$
of indices such that $\subchainleft{i}{j}$ is contained in $D_{r}(c_1)$ and
$\subchainright{i}{j}$ is contained in $D_{r}(c_2)$.
We call such a pair $(i,j)$ of indices a \emph{partition pair} of $D_r(c_1)$ and $D_r(c_2)$.
If $(c_1,c_2)$ is an optimal two-center and $r = \max_{q \in Q}
\{\min\{d(q,c_1), d(q,c_2)\}\}$, then every partition pair of $D_r(c_1)$
and $D_r(c_2)$ is called an \emph{optimal partition pair}.

For a pair $(i,j)$ of indices, an \emph{optimal $(i,j)$-restricted two-center}
is defined as a pair of points $(c_1,c_2)$ that minimizes $r>0$ satisfying
$\subchainleft{i}{j} \subset D_r(c_1)$,
$\subchainright{i}{j} \subset D_r(c_2)$, and $Q \subset D_r(c_1) \cup D_r(c_2)$.
Let $\radrestricted{i}{j}$ be the \emph{radius} of an optimal $(i,j)$-restricted two-center,
defined to be the infimum of all values $r$
satisfying $\subchainleft{i}{j} \subset D_r(c_1)$,
$\subchainright{i}{j} \subset D_r(c_2)$, and $Q \subset D_r(c_1) \cup D_r(c_2)$ for
some $(i,j)$-restricted two-center $(c_1,c_2)$.
Obviously, an optimal two-center of $Q$ is
an optimal $(i, j)$-restricted two-center of $Q$ for some pair $(i, j)$.

In this paper, we give an algorithm for computing an optimal two-center
of $Q$ with respect to $P$.  The overall algorithm is described in
Section~\ref{sec:algorithm}.  As subprocedures, we use the decision
and the optimization algorithms described in
Sections~\ref{sec:decision} and~\ref{sec:opt}, respectively.
The decision algorithm determines whether $r \geq
\radrestricted{i}{j}$ for a given triple $(i,j,r)$ and the optimization
algorithm computes $\radrestricted{i}{j}$ for a given pair $(i,j)$.
While executing the whole algorithm, we call the decision and the
optimization algorithms repeatedly with different inputs.

\section{Decision Algorithm for a Pair of Indices}
\label{sec:decision}

In this section, we present an algorithm that decides whether or not $r \geq
\radrestricted{i}{j}$ given a pair $(i, j)$ and a radius $r\geq 0$.
Note that $r \geq \radrestricted{i}{j}$ if and only if there is a pair $(c_1,c_2)$ of points
in $\CHPQ$ such that $D_r(c_1)$ contains $\subchainleft{i}{j}$, $D_r(c_2)$
contains $\subchainright{i}{j}$ and $D_r(c_1) \cup D_r(c_2)$ contains $Q$.
We call such a pair $(c_1, c_2)$ an \emph{$(i,j,r)$-restricted two-center}.
As discussed above, the set $Q_B$ is partitioned by the pair $(i, j)$
into two subsets, $Q_1=Q_B\cap\subchainleft{i}{j}$ 
and $Q_2=Q_B\cap\subchainright{i}{j}$. 

If $r$ is at least the radius of the smallest geodesic disk containing $\CHPQ$,
which can be computed in time linear to the complexity of $\CHPQ$,
then our decision algorithm surely returns ``yes.''
Another easy case is when $r$ is large enough so that at least one of
the four vertices $v_i$, $v_{i+1}$, $v_j$, $v_{j+1}$ is contained in
both $D_r(c_1)$ and $D_r(c_2)$ for some $(i,j,r)$-restricted two-center $(c_1,c_2)$.
\begin{lemma} \label{lem:decision_easy}
 Given two indices $i, j$ and a radius $r$,
 an $(i,j,r)$-restricted two-center $(c_1, c_2)$ can be computed in
 $O((m+n) \log^2 (m+n))$ time using $O(m+n)$ space,
 provided that $D_r(c_1) \cap D_r(c_2)$ contains one of the four vertices:
 $v_i$, $v_{i+1}$, $v_j$, and $v_{j+1}$.
\end{lemma}
\begin{proof}
Assume without loss of generality that $v_i$ is contained in both $D_r(c_1)$ and $D_r(c_2)$.
Let $Q' := Q \setminus \{v_i \}$.
Then, we observe that $Q'$ can be bipartitioned into $Q'_1$ and $Q'_2$
by a geodesic path $\pi(v_i, w)$ from $v_i$ to some $w\in \bd \CHPQ$
such that $Q'_1 \cup \{v_i\} \subset D_r(c_1)$ and $Q'_2 \cup\{v_i\} \subset D_r(c_2)$.
We thus search the boundary $\bd \CHPQ$ for some $w\in \bd\CHPQ$ that implies
such a bipartition $(Q'_1, Q'_2)$ of $Q'$.

For the purpose, we decompose $\CHPQ$ into triangular cells by extending the shortest path tree rooted at $v_i$
for the vertices of $\CHPQ$ and the points in $Q$ in direction opposite to the root.
There are $O(m+n)$ triangular cells in the resulting planar map $\mathsf{M}$.
All the vertices of $\mathsf{M}$ lie on $\bd \CHPQ$ and
every $q\in Q$ lies on $\pi(v_i, w)$ for some vertex $w$ of $\mathsf{M}$.

We sort the vertices of the cells along $\bd \CHPQ$ in clockwise order from $v_i$,
and then apply a binary search on them to find a vertex $w$
that minimizes the radius of the larger of smallest geodesic disks
containing $(Q' \cap \subpolygon{v_i}{w}) \cup \{v_i\}$
and $(Q' \setminus \subpolygon{v_i}{w}) \cup \{v_i\}$, respectively.
An $(i,j,r)$-restricted two-center then corresponds to the bipartition
obtained in the above binary search.
This takes $O((m+n) \log^2 (m+n))$ time and $O(m+n)$ space.
\end{proof}

Note that
one can also decide in the same time bound
if this is the case where there is an $(i,j,r)$-restricted two-center
$(c_1, c_2)$ such that $v_i \in D_r(c_1) \cap D_r(c_2)$
by running the procedure described in Lemma~\ref{lem:decision_easy}.

In the following, we thus assume that
there is no $(i,j,r)$-restricted two-center $(c_1,c_2)$ such that
any of the four vertices $v_i$, $v_{i+1}$, $v_j$, $v_{j+1}$ is contained in
both $D_r(c_1)$ and $D_r(c_2)$.
This also means that
$D_r(c_1) \cap \{v_i, v_{i+1}, v_j, v_{j+1}\} = \{v_i, v_{j+1}\}$,
$D_r(c_2) \cap \{v_i, v_{i+1}, v_j, v_{j+1}\} = \{v_{i+1}, v_j\}$,
and hence $c_1 \neq c_2$ for any $(i,j,r)$-restricted two-center $(c_1,c_2)$.


\subsection{Intersection of Geodesic Disks and Events}
We first compute the common intersection of the geodesic disks of radius $r$
centered at extreme points on each subchain,
$\subchain{i+1}{j}$ and $\subchain{j+1}{i}$.
That is, we compute $\dintersection_1 = \bigcap_{q\in Q_1} {D_{r}(q)}$ and $\dintersection_2 = \bigcap_{q\in Q_2} {D_{r}(q)}$.
Let $t \in \{1, 2\}$ in the following.
Given the geodesic farthest-point Voronoi diagram of $Q_t$,
the intersection $\dintersection_t$ can be computed in $O(m+n)$ time.
The set $\dintersection_t$ consists of points that are at distance at
most $r$ from all points $Q_t$ and its boundary consists
of parts of boundary arcs of geodesic disks centered at points in $Q_t$
and parts of polygon boundary. Note that any point in a
boundary arc of $\dintersection_t$ has geodesic distance $r$ to some point
in $Q_t$.
We denote the union of the boundary arcs of $\dintersection_t$ by $\boundaryarc_t$.
Note that $r<r^*_{ij}$ if $\dintersection_1 = \emptyset$ or $\dintersection_2 = \emptyset$.
Our decision algorithm returns ``no'' immediately if this is the case.
Otherwise, both $\boundaryarc_1$ and $\boundaryarc_2$ are nonempty because $r$ is smaller than
the radius of a smallest disk containing $\CHPQ$.
Also, if $Q_I = \emptyset$, then our algorithm returns ``yes'' immediately.

\begin{lemma}
  \label{lem:extreme_determinator}
  There is an $(i,j,r)$-restricted two-center $(c_1,c_2)$
  such that $c_1 \in \boundaryarc_1$ and $c_2 \in \boundaryarc_2$,
  provided that $r \geq \radrestricted{i}{j}$.
\end{lemma}
\begin{proof}
  For a fixed $r \geq \radrestricted{i}{j}$, there is at least one
  $(i,j,r)$-restricted two-center.
  Among all such two-centers, let $(c_1,c_2)$ be the one with minimum
  geodesic distance $d(c_1,c_2)$. Recall that $c_1$ and $c_2$ are distinct.
  A point $q \in Q$ with $d(c_1,q)=r$ and $d(c_2,q) > r$
  is called a \emph{determinator} of $D_r(c_1)$.
  We claim that there exists a determinator of $D_r(c_1)$ lying in $Q_1$,
  which implies that $c_1\in \boundaryarc_1$.

\begin{figure}[ht]
  \begin{center}
    \includegraphics[width=0.9\textwidth]{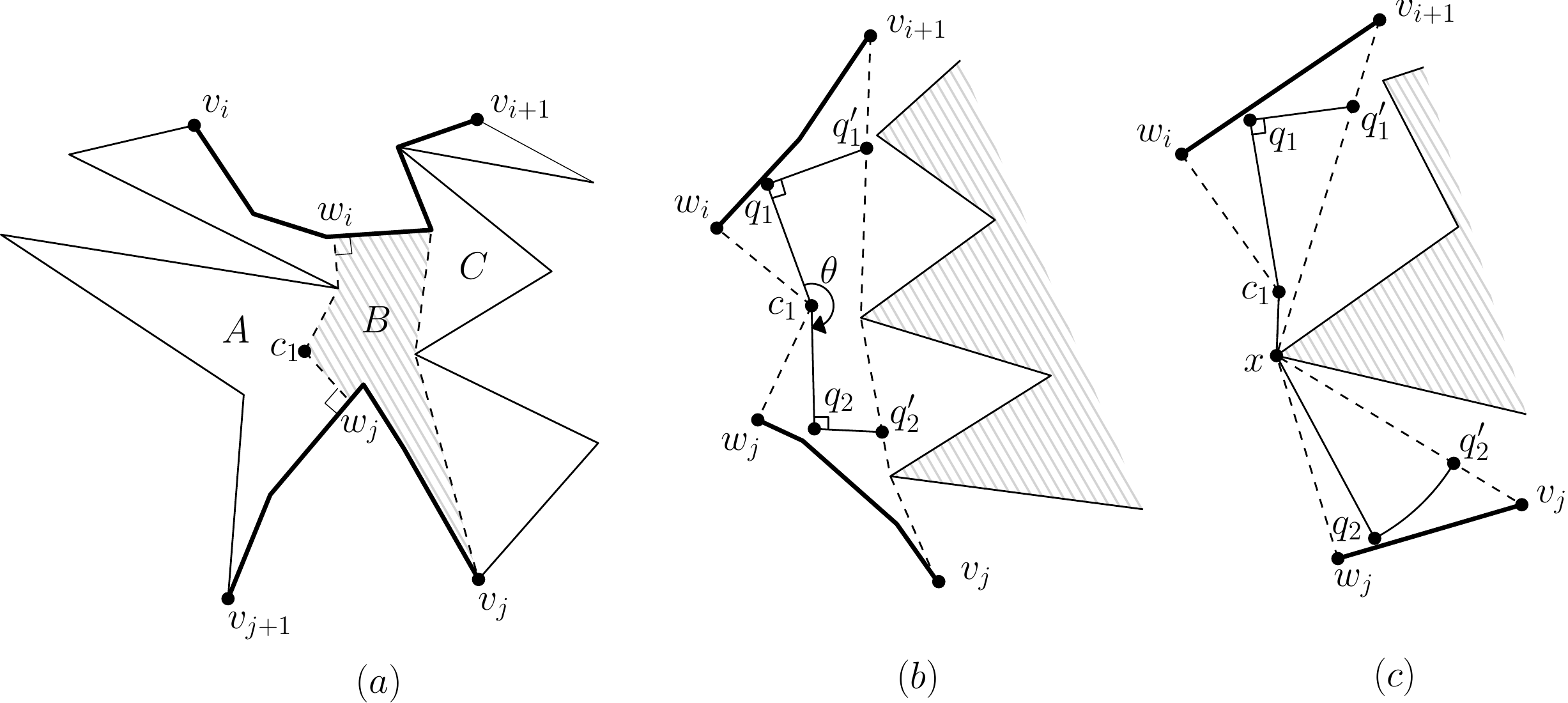}
    \caption{\small (a) The region $A$ is the set of points bounded by $\pi(c_1,w_i)$,
    	$\pi(c_1,w_j)$, and $\subchain{w_j}{w_i}$. The region $C$ is the subpolygon
    	$\subpolygon{v_{i+1}}{v_j}$. The set of points which are not contained either $A$ or $C$
    	is the region $B$.
    	(b) $\pi(c_1,w_i) \cup \pi(c_1,v_{i+1}) =\{c_1\}$ and
    	  $\pi(c_1,w_j) \cap \pi(c_1,v_{j})=\{c_1\}$.
    	(c)
    	$\pi(c_1,w_i) \cup \pi(c_1,v_{i+1}) =\{c_1\}$, but
    	$\pi(c_1,w_j) \cap \pi(c_1,v_{j}) \neq \{c_1\}$.
    \label{fig:boundary_proof}}
  \end{center}
\end{figure}

  Assume to the contrary that no determinators of $D_r(c_1)$ lie in $Q_1$.
  Since $d(c_1,c_2)\leq d(c,c')$ for all $(i,j,r)$-restricted two-centers $(c,c')$,
  there exists at least one determinator of
  $D_r(c_1)$. We subdivide $\CHPQ$ by the following two curves:
  the concatenation of $\pi(c_1,w_i)$ and $\pi(c_1,w_j)$, and
  $\pi(v_{i+1},v_j)$, where $w_i$ and $w_j$ are the points on $\pi(v_i,v_{i+1})$
  and $\pi(v_j,v_{j+1})$ closest to $c_1$, respectively. See
  \figurename~\ref{fig:boundary_proof}(a). These two curves may share some points
  but do not cross.
  Let $A, B,$ and $C$ be the regions of $\CHPQ$ subdivided by the curves,
  as shown in the figure.

  Since every point in $Q$ lying in $C$ must be at distance at most
  $r$ from $c_2$, there is no determinator of $D_r(c_1)$ in $C$.
  We claim that there is no determinator of $D_r(c_1)$ in $A$. Recall that we already assume that
  no determinator of $D_r(c_1)$ lies in $Q_1$.
  If there is a determinator $q$ in $A\setminus Q_1$,
  we show that there is an extreme
  point $q'$ in $Q_1$ such that $d(c_1,q')>r$. Let $x$ be a point on
  $\bd\CHPQ$ from $w_j$ to $w_i$ in clockwise order satisfying
  $q\in\pi(c_1,x)$. Clearly, $x$ lies on one of the following three
  geodesic paths:
  $\pi(v_k,v_{k+1})$ for $v_k,v_{k+1}\in Q_1$, $\pi(v_i,w_i)$, and
  $\pi(w_j,v_{j+1})$. In any case, one endpoint of the geodesic path containing
  $x$ is an extreme point $q'$ of $Q_1$ such that $d(c_1,q')>r$.

  Therefore, the only remaining possibility is that all determinators of
  $D_r(c_1)$ lie in $B$. Consider first the case that there are at least two
  determinators of
  $D_r(c_1)$ in $B$. Let $q_1$ and $q_2$ be the first and last determinators of
  $D_r(c_1)$ that $\pi(c_1,x)$ hits while the point $x$ moves from $w_i$ to
  $w_{j}$ along $\bd\CHPQ$ in clockwise order. Let $\theta$ be
  the clockwise angle from the first segment of $\pi(c_1,q_1)$
  to the first segment of $\pi(c_1,q_2)$ as shown in
  \figurename~\ref{fig:boundary_proof}(b). We claim that there is no pair $(x,y)$
  of points on $\bd\CHPQ$ such that $c_1\in\pi(x,y)$, and
  both $q_1$ and $q_2$ are contained in the region bounded by $\pi(x,y)$
  containing $c_2$. If there is such a pair $(x,y)$, then we can always move
  $c_1$ to $c'_1$ in the direction orthogonal to $\pi(x,y)$ at $c_1$
  infinitesimally such that $D_r(c'_1)$ still contains $Q\cap D_r(c_1)$
  and $d(c'_1,c_2)<d(c_1,c_2)$.
  Therefore we have $\theta\geq\pi$.

  Now we show that $d(v_{i+1},v_j)>2r$, which implies that $d(c_2,v_{i+1})>r$
  or $d(c_2,v_j)>r$, a contradiction.
  There are three possible subcases depending on whether the conditions
  (A) $\pi(c_1,w_i) \cap \pi(c_1,v_{i+1})=\{c_1\}$ and
  (B) $\pi(c_1,w_j) \cap \pi(c_1,v_{j})=\{c_1\}$ hold:
  (1) both (A) and (B) hold,
  (2) either (A) or (B) holds,
  (3) neither (A) nor (B) holds.
  The last case cannot occur since $\theta\geq\pi$ and $c_1$ lies in $A\cup B$
    but not in $\pi(v_{i+1}, v_j)$.
    Thus we show that $d(v_{i+1},v_j)>2r$ for subcases (1) and (2).

  Consider subcase (1).
  The chord passing through $q_1$ and perpendicular to the last segment of $\pi(c_1,q_1)$
  intersects $\pi(v_{i+1},v_j)$. We denote the intersection point
  by $q'_1$. See \figurename~\ref{fig:boundary_proof}(b).
  Similarly, we denote by $q'_2$ the intersection point of $\pi(v_{i+1},v_j)$
  with the chord passing through $q_2$ and perpendicular to the last segment
  of $\pi(c_1,q_2)$.
  Since $\theta\geq\pi$,
  the Euclidean distance between $q_1'$ and $q_2'$ is at least $2r$.
  Therefore, $d(v_{i+1},v_j)=d(v_{i+1},q_1')+d(q_1',q_2')+d(q_2',v_j)>2r$.

  Consider subcase (2).
  Without loss of generality, assume that
  condition (A) holds, but (B) does not hold.
  See \figurename~\ref{fig:boundary_proof}(c).
  Let $x$ be the vertex of $\pi(c_1,v_j)$ next to $c_1$.
  Since (B) does not hold, $\pi(c_1,w_j)$ contains $x$.
  The concatenation of $\pi(v_{i+1},x)$ and $\pi(x,v_j)$ is $\pi(v_{i+1},v_j)$,
  and the concatenation of $\pi(c_1,x)$ and $\pi(x,q_2)$ is $\pi(c_1,q_2)$.
  Moreover, there is the point $q_2'$ in $\pi(x,v_j)$ with
  $d(q_2,x)=d(q_2',x)$ because $d(c_1,v_j)>r=d(c_1,q_2)$.
  Again, let $q_1'$ be the intersection point of $\pi(v_j,v_{i+1})$ and
  the chord passing through $q_1$ and perpendicular to $\pi(c_1,q_1)$.
  The Euclidean distance between $x$ and $q_1'$ is strictly greater than
  $d(x,c_1)+d(c_1,q_1)$ and
  therefore $d(v_{i+1},v_j)=d(v_{i+1},q_1')+d(q_1',x)+d(x,q_2')+d(q_2',v_j)>2r$.

  Now consider the case that there is exactly one determinator $q$ of $D_r(c_1)$
  in $B$.
  Then $c_1$ lies in $\pi(q,c_2)$. Otherwise, we can always reduce both
  $d(c_1,c_2)$ and $d(q,c_1)$
  by moving $c_1$ slightly.
  There is a point $q' \in \pi(c_2,v_{i+1})\cup \pi(c_2,v_j)$ with
  $d(c_2,q)=d(c_2,q')$. Then we have either $d(c_2,v_{i+1})>d(c_2,q)>r$
  or $d(c_2,v_j)>d(c_2,q)>r$, a contradiction.

  In conclusion, if we choose an $(i,j,r)$-restricted two-center $(c_1,c_2)$
  with minimum geodesic distance, at least one of the determinators of
  $D_r(c_1)$ lies in $Q_1$ and at least one of the determinators of
  $D_r(c_2)$ lies in $Q_2$. This implies that
  $c_1 \in \boundaryarc_1$ and $c_2 \in \boundaryarc_2$.
\end{proof}

Let $t \in \{1, 2\}$.
Since the boundary of $\dintersection_t$ is a simple, closed curve,
the clockwise and the counterclockwise directions along $\bd \dintersection_t$
are naturally induced.
We consider the intersection of $D_r(q)$ with $\boundaryarc_t$ 
for each $q\in Q_I$.
Since $\dintersection_t \cap D_r(q)$ is also the intersection of
geodesic disks of the same radius $r$,
by the pseudo-disk property stated in Lemma~\ref{lem:pseudo}, 
the boundary arcs of $\bd D_r(q)$ along $\bd (\dintersection_t \cap D_r(q))$
appears to be consecutive.
This also implies that the union $\boundaryarc_t$ of boundary arcs  
along $\bd \dintersection_t$ 
is divided into two parts, $\boundaryarc_t \cap D_r(q)$ and the rest,
and the arcs contained in each part appear to be consecutive.

We then pick two points $x$ and $x'$ from $\boundaryarc_t \cap D_r(q)$
as follows:
let $x$ and $x'$ be the first points in $\boundaryarc_t \cap D_r(q)$
that we meet while traversing $\bd \dintersection_t$ in clockwise 
and counterclockwise directions, respectively,
from any point $y \in \boundaryarc_t \setminus D_r(q)$.
By Lemma~\ref{lem:pseudo}, the two points $x$ and $x'$
are uniquely defined, regardless of the choice of $y$,
unless $\boundaryarc_t \subseteq D_r(q)$ or 
$\boundaryarc_t \cap D_r(q) = \emptyset$.
We call $x$ and $x'$ the \emph{events} of $q$ 
on $\boundaryarc_t$.
We also associate each event with its defining point $\defpoint{\cdot}$
and a Boolean value $\inout{\cdot} \in \{\iin, \out\}$ as follows: 
$\defpoint{x} = \defpoint{x'} = q$,
$\inout{x} = \iin$, and $\inout{x'} = \out$.
Note that for any $z \in \boundaryarc_t \cap D_r(q)$ with $z \notin \{x, x'\}$,
we meet $x$, $z$, and $x'$ in this order 
during a traversal along $\bd \dintersection_t$ in clockwise direction.

Let $M_t$ be the set of events of all $q\in Q_I$ on $\boundaryarc_t$.
Clearly, the number of events is $|M_t| = O(m)$, and they can be computed as follows:
\begin{lemma}
  \label{lem:time_decision_marking}
  The sets $\boundaryarc_1$, $\boundaryarc_2$, $M_1$, and $M_2$ can be computed
  in $O((m+n) \log^2 (m+n))$ time using $O(m+n)$ space.
\end{lemma}
\begin{proof}
We will make use of known geometric structures, namely,
geodesic (nearest-point) Voronoi diagrams and geodesic farthest-point Voronoi diagrams,
which are counterparts of the standard Voronoi diagrams and farthest-point Voronoi diagrams
with respect to the geodesic distance~\cite{VD, FVD}.
For a set $S$ of $N$ point sites in $P$, both
the geodesic Voronoi diagram $\gvd{S}$ and the geodesic farthest-point Voronoi diagram
$\fvd{S}$ can be computed in $O((n+N) \log (n+N))$ time using $O(n+N)$ space.
They support an $O(\log (n+N))$-time closest- or farthest-site query with respect to
the geodesic distance for any query point in $P$~\cite{VD, FVD}.

We first compute the intersections $\dintersection_1$ and $\dintersection_2$ of geodesic disks
in $O((m+n) \log (m+n))$ time by using the geodesic farthest-point Voronoi diagrams
$\fvd{Q_1}$ and $\fvd{Q_2}$ of $Q_1$ and $Q_2$, respectively.
A cell in $\fvd{Q_1}$ corresponding to a site $t$
consists of the points $p \in \CHPQ$ such that $t$ is the site
farthest from $p$ among all sites.  A \emph{refined} cell in $\fvd{Q_1}$
with site $t$ is obtained by further subdividing the cell of site $t$
such that all points in the same refined cell have the
combinatorially equivalent shortest paths from their farthest point $t$.
While constructing
$\fvd{Q_1}$ and $\fvd{Q_2}$, for each refined cell,
we store the information about the site $t$ of the refined cell
and the last vertex of
$\pi(t,p)$ for a point $p$ in the refined cell.

We first show that the number of a boundary arc of $\dintersection_1$ is $O(m+n)$.
Let $s$ be a boundary arc of $\dintersection_1$. The center
$c_s$ of the geodesic disk containing $s$ on its boundary lies in
$\subchain{v_{j+1}}{v_i}$.  Note that $c_s$ is unique by the general
  position assumption.
Every geodesic disk whose center is a vertex
in $\subchain{v_{j+1}}{v_i} \setminus \{c_s\}$ contains $s$ in its
interior.  This means that, for any point $x \in s$, the farthest
point from $x$ in $\subchain{v_{j+1}}{v_i}$ is $c_s$.  Moreover, the
geodesic paths from the center $c_s$ to points on the boundary arc
$s$ are combinatorially equivalent.  Thus each boundary arc $s$ is
contained in the refined cell of the farthest-point geodesic Voronoi
diagram whose site is $c_s$.  Moreover, each endpoint of the
boundary arc lies in either the boundary of the cell containing it
or the boundary of $P$.	
The number of boundary arcs at least one of whose endpoints lies on
$\partial P$ is $O(n)$ and the number of boundary arcs none of whose
endpoints lies on $\partial P$ is $O(m+n)$ by the fact that the size
of the geodesic farthest-point Voronoi diagram is $O(m+n)$.

The sets $\boundaryarc_1$ and $\boundaryarc_2$ are subsets of $\bd\dintersection_1$ and $\bd\dintersection_2$,
respectively, so can be extracted in $O(m+n)$ time.
Also, note that $\boundaryarc_1$ and $\boundaryarc_2$ consist of $O(m+n)$ arcs.

In order to compute the sets $M_1$ and $M_2$ of events on $\boundaryarc_1$
and $\boundaryarc_2$, we compute the geodesic nearest-site
and the geodesic farthest-point Voronoi diagrams of the endpoints of the arcs
in $\boundaryarc_t$ for each $t=1,2$.
By an abuse of notation, we denote these diagrams by $\gvd{\boundaryarc_t}$
and $\fvd{\boundaryarc_t}$, respectively.
Since $\boundaryarc_t$ consists of $O(m)$ boundary arcs of $\dintersection_t$,
the diagrams $\gvd{\boundaryarc_t}$
and $\fvd{\boundaryarc_t}$ can be constructed in $O((m+n) \log (m+n))$ time,
and support an $O(\log (m+n))$-time closest- or farthest-site query.

Consider a fixed $q\in Q_I$.
In the following, we show that
$\bd D_r(q) \cap \boundaryarc_t$ can be computed
in $O(\log (m+n) \log n)$ time, using $\gvd{\boundaryarc_t}$,
and $\fvd{\boundaryarc_t}$.
Thus, the total time spent in computing $M_t$ is bounded by $O(m \log(m+n)\log n)$,
and the lemma follows.

By Lemma~\ref{lem:pseudo}, the size of $\boundaryarc_t \cap \bd
  D_r(q)$ is at most two.
  Assume that $\boundaryarc_t \cap \bd D_r(q)$ is not
  empty, and let $x$ and $x'$ be two points in the set. We compute the arcs $a$ and $a'$ of $A_t$
  containing $x$ and $x'$, respectively, by applying a binary search on the endpoints of the  arcs of $A_t$.
  To do this, we need two points $x_c$ and $x_f$ in $A_t$ with $x_c\in D_r(q)$ and $x_f\notin D_r(q)$.
  We can compute $x_c$ and $x_f$ as follows. First, we find the farthest and closest endpoint of the  arcs of $A_t$ from $q$ using $\fvd{A_t}$ and $\gvd{A_t}$ in $O(\log (n+m))$ time.
  Then we consider the four  arcs of $A_t$ incident to such endpoints. Due to Lemma~\ref{lem:pseudo},
  we can find two points $x_c \in D_r(q)$ and $x_f \in D_r(q)$ in the four  arcs.
Without loss of generality, we assume that
$x$ comes before $x'$ as we traverse $\boundaryarc_t$ from $x_c$
in clockwise order.
Note that   every point of $\boundaryarc_t$ from $x_c$ to $x$ in clockwise order
is contained in $D_r(q)$, while
every point of $\boundaryarc_t$ from $x$ to $x_f$ is not
contained in $D_r(q)$.

Exploiting this property, we apply a binary search on the endpoints of $\boundaryarc_t$
to find the  arc $a \in \boundaryarc_t$ containing $x$.
Let $x_{\mathrm{med}}$ be the median of the endpoints of $\boundaryarc_t$
  from $x_c$ to $x_f$ in clockwise order.
  If $d(x_{\mathrm{med}},q) > r$, then $x$ lies between $x_c$
  and $x_{\mathrm{med}}$. Otherwise,
  $x$ lies between $x_{\mathrm{med}}$ and $x_f$.
  Thus, in $O(\log(m+n))$ iterations, we find the  arc $a$ containing $x$.
  In each iteration, we compute the geodesic distance of two points,
  which takes $O(\log n)$ time.
  In total, the  arc $a$ containing $x$ can be found in
  $O(\log n \log(m+n))$ time.
  Similarly, we can compute the  arc $a'$ containing $x'$.

  Now, we find the exact location of $x$ on the  arc $a$.
  Let $\alpha_1$ and $\alpha_2$ be the endpoints of $a$.
  We can compute the point $\alpha$ such that the maximal common path of $\pi(q,\alpha_1)$ and $\pi(q,\alpha_2)$
  is $\pi(q,\alpha)$ in $O(\log n)$ time using the data structure of size $O(n)$ given by Guibas and Hershberger~\cite{shortest-path}. For any point $p\in a$, the path $\pi(q,a)$ is the composition of $\pi(q,\alpha)$, $\pi(\alpha,\beta)$ and the line segment $\overline{\beta p}$ for some vertex $\beta$
  of $\pi(\alpha,\alpha_1)\cup\pi(\alpha,\alpha_2)$. Once we obtain $\beta$, we can compute $x$ in constant time.
  To obtain $\beta$, we apply a binary search on the vertices of $\pi(\alpha,\alpha_1)\cup\pi(\alpha,\alpha_2)$
  as we did for computing $a$. Specifically, imagine that we extend all edges of $\pi(\alpha,\alpha_1)\cup\pi(\alpha,\alpha_2)$ towards $a$. The extensions subdivide $a$ into $O(n)$ smaller arcs.
  We apply a binary search on the endpoints of the smaller arcs in $O(\log^2 n)$ time as we did before.
  Then we obtain the smaller arc containing $x$, and thus we obtain $\beta$ in $O(\log^2 n)$ time.

  In this way, we obtain $x$ and $x'$ in $O(\log(n+m)\log n)$ time for each point in $Q$.
  Therefore, the sets $M_1$ and $M_2$ of events on $\boundaryarc_1$
  and $\boundaryarc_2$ can be computed in $O(m\log(n+m)\log n)$ time in total.
    The space used above is also bounded by $O(m+n)$.
\end{proof}

The sets $M_1$ and $M_2$ of events indeed play an important role for our decision algorithm.
\begin{lemma}
	\label{lem:event}
    Suppose that both $M_1$ and $M_2$ are nonempty.
	Then there is an $(i,j,r)$-restricted two-center $(c_1,c_2)$ such that
	$c_1 \in M_1$ and $c_2 \in M_2$, if $r \geq \radrestricted{i}{j}$.
\end{lemma}
\begin{proof}
  If $r \geq \radrestricted{i}{j}$, there is an $(i,j,r)$-restricted two-center $(c_1,c_2)$
  with $c_1 \in \boundaryarc_1$ and $c_2 \in
  \boundaryarc_2$ by Lemma~\ref{lem:extreme_determinator}.
  We find an event $x\in M_1$ such that $(x, c_2)$ is an $(i,j,r)$-restricted two-center.
  If $c_1 \in M_1$, then we are done by setting $x=c_1$. Otherwise, let $x$ be the first event of $M_1$ that is
  encountered while traversing $\boundaryarc_1$ from $c_1$
  in any direction. 
  By the construction of $M_1$, 
  we have $D_r(c_1) \cap Q = D_r(x) \cap Q$.
  Thus, $(x, c_2)$ is still an $(i,j,r)$-restricted two-center.
  Similarly, we can find an event $y\in M_2$ such that $(c_1,y)$ is
  an $(i,j,r)$-restricted two center. This also implies that $(x,y)$ is an $(i,j,r)$-restricted
  two center.
\end{proof}

If both event sets, $M_1$ and $M_2$, are empty, then
for every $q\in Q_I$,
$\bd D_r(q)$ intersects neither $\boundaryarc_1$ nor $\boundaryarc_2$.
This means that either $\boundaryarc_t \subseteq D_r(q)$ or $\boundaryarc_t \cap D_r(q) = \emptyset$
for each $t\in \{1, 2\}$.
If $\boundaryarc_t \subseteq D_r(q)$ for some $t \in \{1, 2\}$,
then $q$ is contained in a geodesic disk of radius $r$ centered at any point on $\boundaryarc_t$.
Otherwise, if $\boundaryarc_t \cap D_r(q) = \emptyset$ for all $t\in \{1, 2\}$,
then $q$ cannot be contained any such disk centered at a point on
$\boundaryarc_1$ or $\boundaryarc_2$.
So, our decision algorithm should return ``no'' if there is $q\in Q_I$ in the latter case
by Lemma~\ref{lem:extreme_determinator};
while it returns ``yes'' if this is the former case for all $q\in Q_I$.

If one of them is empty, say $M_1 = \emptyset$, and
$(c_1, c_2)$ is an $(i, j, r)$-restricted two-center, then we must
have $Q_I \subset D_r(c_2)$.  Thus, this case can be handled by
computing the smallest geodesic disk containing $Q_2 \cup Q_I$ and
testing if its radius is at most $r$.

Hence, in the following, we assume that both $M_1$ and $M_2$ are nonempty.
Then, by Lemma~\ref{lem:event}, we can decide if $r \geq \radrestricted{i}{j}$
by finding a pair $(c_1,c_2)$ of points such that
$c_1 \in M_1, c_2 \in M_2$ and $Q_I \subset D_r(c_1) \cup D_r(c_2)$.
For the purpose, we traverse $\boundaryarc_1$ and $\boundaryarc_2$
simultaneously by handling the events in $M_1$ and $M_2$ in this traversed order
from proper reference points on $\boundaryarc_1$ and $\boundaryarc_2$, respectively.

Our \emph{reference point} $o_t$ on $\bd \dintersection_t$ for each $t\in \{1, 2\}$
should satisfy the following condition:
$o_t \in \boundaryarc_t$ and, for every $q\in Q_I$,
either
\begin{enumerate}[(i)] \denseitems
\item $\bd D_r(q) \cap \boundaryarc_t = \emptyset$,
\item $\boundaryarc_t \subseteq D_r(q)$, or
\item we meet $o_t$, $x$, and $x'$ in this order when traversing $\bd \dintersection_t$ in clockwise direction, possibly being $o_t = x$,
    where $x$ and $x'$ are two events of $q$ on $\boundaryarc_t$ such that
    $\inout{x} = \iin$ and $\inout{x'} = \out$.
\end{enumerate}
Such reference points $o_1$ and $o_2$ can be found in $O(m+n)$ time.
\begin{lemma} \label{lem:reference}
 For each $t \in\{1, 2\}$,
 such a reference point $o_t \in \boundaryarc_t$ on $\bd \dintersection_t$ exists,
 and can be found in $O(m+n)$ time using $O(m+n)$ space.
\end{lemma}
 \begin{proof}
If $Q_I = \emptyset$, it is trivial that any point in $\boundaryarc_t$ can be chosen
as such a reference point $o_t$.
We thus assume $Q_I \neq \emptyset$.
For each $t\in \{1, 2\}$, let $Q'_t$ denote the set of extreme points $q$ in $Q_t$ such that
a boundary arc of $D_r(q)$ appears on $\boundaryarc_t$.
Since $\boundaryarc_t\neq\emptyset$, $Q'_t$ is nonempty.

Assume first that $|Q'_1| = 1$, and $Q'_1 = \{v \}$.
Then we pick the first point in $\boundaryarc_t$
when traversing $\bd \dintersection_1$ in clockwise direction from $v$,
and keep it as the reference point $o_1$.
(Note that $v \in \bd \dintersection_1$ since we regard $\CHPQ$ as the input polygon.)
Suppose to the contrary that the condition for a reference point is
violated for some $q \in Q_I$ with $\bd D_r(q) \cap \boundaryarc_1 \neq \emptyset$.
Then, we have $x \neq o_1$ and we meet $o_1$, $x'$, and $x$ in this order
when traversing $\bd \dintersection_1$ in clockwise direction,
where $x$ and $x'$ are the events of $q$ on $\boundaryarc_1$
with $\inout{x} = \iin$ and $\inout{x'} = \out$.
This implies that there is $y\in \boundaryarc_1$
such that $y \notin D_r(q)$.
Since the three points $x$, $x'$, $y$ lie on a boundary arc of $D_r(v)$,
we have $d(v, x) = d(v, x') = d(v, y) = r$.
This implies that the geodesic Voronoi diagram of three points $\{x, x', y\}$
has a unique vertex at $v$.
On the other hand, observe that $d(q, y) > r \geq \max\{ d(q, x), d(q, x')\}$
since $y \notin D_r(q)$ and $x, x' \in \bd D_r(q)$.
So, $q$ lies on the region of $x$ or of $x'$ in the diagram.
If $q$ lies on the region of $x$, then
the shortest path from $q$ to $x'$ must passes through the region of $y$
since $v$ is the unique vertex of the diagram.
This implies that $d(q, x') > d(q, y)$, a contradiction.
The other case where $q$ lies on the region of $x'$ can be handled similarly.

Next, we assume that $|Q'_1| \geq 2$.
Our proof makes use of a known property of the intersections of geodesic disks,
which follows from the basic properties of the geodesic farthest-point Voronoi diagrams
and Corollary 2.7.4 of Aronov et al.~\cite{FVD}.
\begin{quote}
\textit{%
 (*) Let $C$ be a finite set of points in a simple polygon $P$ and
	$I := \bigcap_{c \in C} D_r(c)$ for $r>0$.
	Each boundary arc of $I$ is part of $\bd D_r(c)$ for some extreme point $c\in C$,
    that is, a point $c$ that lies on the boundary of the geodesic convex hull of $C$.
    Let $C'$ be the set of extreme points $c$ of $C$ such that
    a boundary arc of $D_r(c)$ appears on $\bd I$.
	Then, the order of points in $C'$
	along the boundary of the geodesic convex hull of $C$ is the same as
	the order of their boundary arcs along $\bd I$.
	}
\end{quote}

For a point $q \in Q_I \cap \subpolygon{v_{j+1}}{v_i}$,
we have $\boundaryarc_1 \subset D_r(q)$
since $\subpolygon{v_{j+1}}{v_i}$ is the geodesic convex hull of $Q_1$
and hence $q$ is not an extreme point of $Q_1 \cup \{q\}$.
Thus, we do not need to consider such points $q \in Q_I \cap \subpolygon{v_{j+1}}{v_i}$
in this proof.
On the other hand, a point $q \in Q_I \cap \subpolygon{v_i}{v_{j+1}}$
is an extreme point of the geodesic convex hull of $Q_1\cup\{q\}$.
Moreover, its two neighboring extreme points on the geodesic convex hull of
$Q_1\cup \{q\}$ are always $v_i$ and $v_{j+1}$.

We choose two points $u, w \in Q'_1$ as follows:
If $v_i \in Q'_1$, then $u = v_i$; otherwise,
$u$ is the first point in $Q'_1$ in counterclockwise direction from $v_i$ along
$\bd \subpolygon{v_{j+1}}{v_i}$.
Similarly, if $v_{j+1} \in Q'_1$, then $w = v_{j+1}$; otherwise,
$w$ is the first point in $Q'_1$ in clockwise direction from $v_{j+1}$ along
$\bd \subpolygon{v_{j+1}}{v_i}$.
By our choice of $u$ and $w$, note that $u$ and $w$ are consecutive in $Q'_1$
along $\bd \subpolygon{v_{j+1}}{v_i}$.
By the above property (*),
the boundary arcs that belong to $\bd D_r(u)$ and $\bd D_r(w)$ appear to be consecutive
along $\bd \dintersection_1$.
That is, there is no $z \in Q'_1$ with $z \notin \{u, w\}$ such that
three boundary arcs that belong to $\bd D_r(u)$, $\bd D_r(z)$, and $\bd D_r(w)$ appear in this order
when we walk along $\bd \dintersection_1$ from any point in $D_r(x) \cap \bd \dintersection_1$
in clockwise direction.
We now consider the set $X$ of endpoints of boundary arcs of $\dintersection_1$
that belong to $\bd D_r(u)$.
We then choose the first point in $X$ in counterclockwise direction
from any point on $\bd D_r(w) \cap \boundaryarc_1$ along $\bd \dintersection_1$,
and denote it by $o_1$.

Next we show that the chosen point $o_1$ is indeed a reference point on $\bd \dintersection_1$.
Let $q \in Q_I\cap\subpolygon{v_i}{v_{j+1}}$.
If $\boundaryarc_1 \subset D_r(q)$ or $\bd D_r(q) \cap \boundaryarc_1 = \emptyset$,
then we are done.
Otherwise, a boundary arc $\gamma$ of $D_r(q)$ appears on $\bd (\dintersection_1 \cap D_r(q))$
between boundary arcs from $D_r(x)$ and $D_r(y)$,
by property (*) and the fact that the neighboring extreme points of $q$ on the geodesic convex hull of
$Q_1\cup \{q\}$ are $v_i$ and $v_{j+1}$.
This implies that $o_1$ avoids the interior of $D_r(q)$,
and thus we get this order $o_1$, $x$, $x'$ when traversing $\dintersection_1$ in clockwise direction,
where $x$ and $x'$ are events of $q$ on $\boundaryarc_1$
with $\inout{x} = \iin$ and $\inout{x'} = \out$.
This shows that $o_1$ is a reference point on $\boundaryarc_1$.
Note that our reference point $o_1$ is defined uniquely by the above procedure.

Similarly, we can show that the reference
point $o_2$ on $\boundaryarc_2$ exists and is defined uniquely.
In order to find $o_t$, it suffices to traverse $\boundaryarc_t$, $\dintersection_t$, and $Q_t$
taking $O(m+n)$ time and space for $t \in \{1,2\}$.
\end{proof}

Using the reference point $o_t$, we define an order $\prec_t$ on $\dintersection_t$ as follows.
We write $x \prec_t y$ for two points $x, y \in \bd \dintersection_t$ if $x$ comes before $y$
as we traverse $\bd \dintersection_1$ in clockwise order from the reference point $o_t$.
We also write $x \preceq_t y$ if either $x = y$ or $x\prec_t y$.
Since $\boundaryarc_t \subseteq \dintersection_t$, the order $\prec_t$ on $\boundaryarc_t$ is naturally inherited.
Note that if $x$ and  $x'$ are two events of $q\in Q_I$ on $\boundaryarc_t$ with $x\preceq_t x'$,
then we have $\inout{x} = \iin$, $\inout{x'} = \out$, and
$x \preceq_t y \preceq_t x'$  for any $y \in D_r(q) \cap \boundaryarc_t$.

\subsection{Traversing \texorpdfstring{$\boundaryarc_1$}{A1} and \texorpdfstring{$\boundaryarc_2$}{A2} by Scanning Events}
\label{sec:traversing}

As a preprocessing,
we sort the events in $M_1$ and $M_2$ with respect to the orders $\prec_1$ and $\prec_2$,
respectively.
We scan $M_1$ once by moving a pointer $\centerl$ from the reference point $o_1$
in clockwise order.
We also scan $M_2$ from
the reference point $o_2$ of $\bd \dintersection_2$ by moving one pointer $\centerra$ in
clockwise order and another pointer $\centerrb$ in counterclockwise
order at the same time.
We continue to scan and handle the events until
$\centerl$ points to the last event of $M_1$ or
$\centerra$ and $\centerrb$ point at the same event of $M_2$.
We often regard the three pointers as events which they point to.
For example, we write $D_r(\centerra)$ to indicate
the set of points whose geodesic distance from the event in $M_2$
which $\centerra$ points to is at most $r$.

Whenever we handle an event, we apply two operations,
which we call \dcs and \upd.
We maintain the sets $D_r(\centerl) \cap Q_I$, $D_r(\centerra) \cap Q_I$,
and $D_r(\centerrb) \cap Q_I$.
Operation \upd updates the sets, and operation \dcs checks whether $Q_I \subset
D_r(\centerl) \cup D_r(\centerra)$ or
$Q_I \subset D_r(\centerl) \cup D_r(\centerrb)$.

In the following, we describe
how to handle the events in $M_1 \cup M_2$, and
how the two operations work.

\paragraph{Handling Events in \texorpdfstring{$M_1 \cup M_2$}{M1 cup M2}.}
We move the three pointers $\centerl$, $\centerra$ and $\centerrb$ as
follows.  First, we scan $M_1$ from the current $\centerl$
in clockwise order until we reach an event $x$
with $\inout{x} = \out$.  We set $\centerl$ to $x$.
If $D_r(\centerra)$ does not contain $\defpoint{x}$,
then we scan $M_2$ from $\centerra$ in clockwise order until we reach
the event $y$ with $\defpoint{x}=\defpoint{y}$, and set $\centerra$ to $y$.
If $D_r(\centerrb)$ does not contain $\defpoint{x}$, then we also scan
$M_2$ from $\centerrb$ in counterclockwise order until we reach the event $y'$ with
$\defpoint{x}=\defpoint{y'}$, and set $\centerrb$ to $y'$.
We check whether $Q_I \subset D_r(\centerl) \cup D_r(c_2^s)$ for $s=c,cc$.
If yes, we stop traversing and return as a
solution $(\centerl,\centerra)$ or $(\centerl,\centerrb)$ accordingly.
Otherwise, we repeat the scan above and
check
whether $Q_I \subset D_r(\centerl) \cup D_r(\centerra)$ or
$Q_I \subset D_r(\centerl) \cup D_r(\centerrb)$
for events in $M_1 \cup M_2$ encountered during the scan.
If this test passes at some event,
we stop traversing and return as a
solution $(\centerl,\centerra)$ or $(\centerl,\centerrb)$ accordingly.
If the pointer $\centerl$ goes back to the reference point
or $\centerra, \centerrb$ meet each other, our decision algorithm returns ``no''.
Clearly, this algorithm terminates and we consider $O(m)$ event
points in total.  If both \upd and \dcs take constant time, the total running
time for this step is $O(m)$.

\paragraph{Operations \dcs and \upd.}

To apply \dcs and \upd in constant time, we use five arrays for the points in $Q_I$.
Each element of the arrays is a Boolean value corresponding to each point in
$Q_I$.
For the first array, each element indicates whether $D_r(\centerl)$ contains
its corresponding point in $Q_I$.  Similarly, the second and the third
arrays have Boolean values for $D_r(\centerra)$ and $D_r(\centerrb)$,
respectively.
Each element of the remaining two arrays indicates whether its corresponding
point in $Q_I$ is contained in $D_r(\centerl) \cup D_r(\centerra)$ and
$D_r(\centerl)\cup D_r(\centerrb)$, respectively.
In addition to the five arrays, we
also maintain five counters that represent the number of points of $Q_I$
contained in each of the following five sets:
$D_r(\centerl)$, $D_r(\centerra)$, $D_r(\centerrb)$,
$D_r(\centerl) \cup D_r(\centerra)$, and $D_r(\centerl) \cup D_r(\centerrb)$.

At the reference points, we initialize the five arrays and the five counters
in $O(m)$ time.
For \dcs, we just check whether the number of points contained in
either $D_r(\centerl) \cup D_r(\centerra)$ or
$D_r(\centerl) \cup D_r(\centerrb)$ is equal to the number of points in $Q_I$,
which takes constant time.
To apply \upd when
$\centerl$ reaches an event $x \in M_1$ with $\defpoint{x}=q$,
we first change Boolean values of the elements in the arrays assigned for
$D_r(\centerl)$, $D_r(\centerl) \cup D_r(\centerra)$ and $D_r(\centerl) \cup D_r(\centerrb)$
according to $\inout{x}$.
When we change Boolean values, we also update the counters of the sets accordingly.
These procedures can be done in constant time.

We are now ready to conclude this subsection as follows.
\begin{theorem}
  \label{thm:decision_algorithm}
  Given a pair $(i, j)$ and
  a nonnegative real $r$, our decision algorithm decides whether or
  not $r \geq \radrestricted{i}{j}$ correctly in
  $O((m+n) \log^2(m+n) )$ time using $O(m+n)$ space.  Moreover, it
  also returns an $(i,j,r)$-restricted two-center, if
  $r \geq \radrestricted{i}{j}$.
\end{theorem}
\begin{proof}
	First, we prove the correctness of the algorithm.
	If $r \geq \radrestricted{i}{j}$,
	among all $(i,j,r)$-restricted two-centers,
	we choose a two-center such that
	$c_1' \in \boundaryarc_1$, $c_2' \in \boundaryarc_2$, and
	no pair $(c_x,c_y)$ with $c_1' \prec_1 c_x$
	is an $(i,j,r)$-restricted two-center.
	Then $c_1'$ is an event $x$ of $M_1$ with $\inout{x}=\out$.
	Moreover, if $D_r(c_2')$
	contains $\defpoint{x}$, there is an $(i,j,r)$-restricted two-center $(c_x,c_2')$
    with $c_1' \prec_1 c_x$.  Thus, $c_2' \notin
	D_r(\defpoint{x})$.
	
	Let $y^{\mathrm{c}}$ and
	$y^{\mathrm{cc}}$ be the events for $\centerra$ and $\centerrb$ in $M_2$
	right before we handle the event $x$ in $\boundaryarc_1$, respectively.
	Then there is an event $x' \prec_1 x$
	such that $\inout{x'}=$ $\out$ and $\defpoint{y^{\mathrm{c}}}=
	\defpoint{x'}$.
	After we handle $x'$ in $M_1$,
	$D_r(\centerl)$ does not contain $\defpoint{x'}$ any more.
	Thus $D_r(c_2')$ contains $\defpoint{x'}$.
	This implies that $y^{\mathrm{c}} \prec_2 c_2'$.
	Similarly, we have $c_2' \prec_2 y^{\mathrm{cc}}$.
	The set of
	the events between $y^{\mathrm{c}}$ and $y^{\mathrm{cc}}$ which
	are not contained in $D_r(\defpoint{x})$ are exactly the events our
	algorithm handles right after $\centerl$ reaches $x$.
	Thus, our algorithm always finds $(c_1',c_2')$,
	which is an $(i,j,r)$-restricted two-center.
	
	If $r < \radrestricted{i}{j}$, then
    our algorithm cannot find any $(i,j,r)$-restricted two-center $(c_1,c_2)$
    before $\centerl$ goes back to the reference
    point of $\boundaryarc_1$, or $\centerra$ and $\centerrb$ meet
    each other.
    Lemma~\ref{lem:event} implies that there is no $(i,j,r)$-restricted two-center
    and the decision algorithm correctly reports that $r < \radrestricted{i}{j}$
    in this case.

    For the time complexity,
    our decision algorithm first checks 
    whether or not $r$ is at least the radius of the smallest geodesic disk containing $\CHPQ$,
    and whether or not at least one of
the four vertices $v_i$, $v_{i+1}$, $v_j$, $v_{j+1}$ is contained in
both $D_r(c_1)$ and $D_r(c_2)$ for some $(i,j,r)$-restricted two-center $(c_1,c_2)$.
    This takes $O((m+n) \log^2 (m+n))$ time as discussed in Lemma~\ref{lem:decision_easy}.
    Then, the algorithm computes $\boundaryarc_t$ and $M_t$ for each $t=1,2$
    in $O((m+n) \log^2 (m+n))$ time as shown in
    Lemma~\ref{lem:time_decision_marking}.
    Once $\boundaryarc_t$ and $M_t$ are ready,
    we scan $M_1$ and $M_2$ as described above.
    Scanning $M_1$ and $M_2$ is done in $O(m)$ time
    since handling each event and performing each operation can be done in $O(1)$ time
    as discussed above.
    Thus, the claimed time bound is implied.

    The space complexity is bounded by $O(m+n)$
    because the procedures described in Lemmas~\ref{lem:decision_easy} and~\ref{lem:time_decision_marking} take $O(m+n)$ space.
\end{proof}

\section{Optimization Algorithm for a Pair of Indices}
\label{sec:opt}
In this section, we present an optimization algorithm for a given
pair $(i, j)$ that computes $\radrestricted{i}{j}$ and an
optimal $(i, j)$-restricted two-center.

Our optimization algorithm works with a left-open and right-closed interval $(r_L, r_U]$,
called an \emph{assistant interval}, which will be given also as part of input.
An assistant interval $(r_L,r_U]$ should satisfy the following condition:
$r^* \in (r_L,r_U]$ and
the combinatorial structure of $\bd D_r(q)$ for each $q\in Q$ remains the same
for all $r\in (r_L,r_U]$, where $r^* = \min_{i,j} \radrestricted{i}{j}$ denotes
the radius of an optimal two-center of $Q$.
We will see later in Lemma~\ref{lem:assistant_interval}
how we obtain such an assistant interval efficiently.
The algorithm returns the value $\radrestricted{i}{j}$ if
$\radrestricted{i}{j} \leq r_U$; otherwise, it just reports that
$\radrestricted{i}{j} > r_U$.  The latter case means that $(i, j)$ is
not an optimal partition pair, as we have assured that
$\radrestricted{i}{j} > r_U \geq r^*$.  Testing whether
$\radrestricted{i}{j} > r_U$ or $\radrestricted{i}{j} \leq r_U$ can be
done by running the decision algorithm with input $(i, j, r_U)$.  In
the following, we thus assume that
$\radrestricted{i}{j} \in (r_L,r_U]$, and search for
$\radrestricted{i}{j}$ in the assistant interval $(r_L,r_U]$.

As in the decision algorithm, we consider the intersection of geodesic disks and events
on the arcs of $\boundaryarc_t$.
For each $r\in (r_L,r_U]$ and each $t\in \{1, 2\}$, let $\dintersection_t(r) := \bigcap_{q\in Q_t} D_r(q)$,
and $\boundaryarc_t(r)$ be the union of boundary arcs of $\dintersection_t(r)$.
Also, let $M_t(r)$ be the set of events of each $q\in Q_I$, if any, on $\boundaryarc_t(r)$,
as defined in Section~\ref{sec:decision}.
Here, we identify each event $x \in M_t(r)$ by a pair $x=(\defpoint{x}, \inout{x})$,
not by its exact position on $\boundaryarc_t$.
Note that the set $M_t(r)$ and the combinatorial structure of $\bd \dintersection_t(r)$ may not be constant
over $r\in (r_L,r_U]$.
In order to fix them, we narrow down the assistant interval $(r_L,r_U]$
to $(\rho_L,\rho_U]$ as follows.
\begin{lemma} \label{lem:narrowed_interval}
 One can find an interval $(\rho_L,\rho_U] \subseteq (r_L,r_U]$ containing $\radrestricted{i}{j}$ in $O((m+n) \log^3(m+n))$ time
 such that the combinatorial structure of each of the following remains the same
 over $r\in (\rho_L,\rho_U]$: $\bd \dintersection_t(r)$ and $M_t(r)$ for $t=1,2$.
\end{lemma}
\begin{proof}
  We fix the combinatorial structure of $\bd \dintersection_t(r)$ for
  each $t \in \{1, 2\}$ as follows.  We compute the geodesic
  farthest-point Voronoi diagrams $\fvd{Q_1}$ and $\fvd{Q_2}$ of $Q_1$
  and $Q_2$, respectively.  Then for each vertex $v$ of $\fvd{Q_1}$ or
  $\fvd{Q_2}$, we compute the geodesic distance between $v$ and its
  farthest neighbor in $Q_1$ or $Q_2$.  We sort the $O(m+n)$ geodesic
  distances and apply a binary search on them to find an interval
  $(r'_L,r'_U]$ using the decision algorithm such that
  $r'_L \leq \radrestricted{i}{j} \leq r'_U$. Then for any $r\in (r'_L,r'_U]$,
  the combinatorial structure of $\bd \dintersection_t(r)$ for
  each $t \in \{1, 2\}$ remains the same.

Next, we fix $M_1(r)$ and $M_2(r)$ by using $\fvd{Q_1}$ and $\fvd{Q_2}$.
For each $q\in Q_I$ and $t\in\{1,2\}$, let $\rho_t(q)$ be half the geodesic distance
from $q$ and its farthest neighbor in $Q_t$.
The value $\rho_t(q)$ can be computed in $O(\log (m+n))$ time.
As $r$ increases from $r'_L$ to $r'_U$ continuously,
$M_t(r)$ changes only when $r = \rho_t(q)$ for some $q\in Q_I$
since $q$ lies on $\bd \dintersection_t(2\rho_t(q))$ by definition and thus
$\dintersection_t(\rho_t(q))$ touches $D_{\rho_t(q)}(q)$.
Thus, we gather all these values $\rho_t(q)$ with $\rho_t(q) \in (r'_L,r'_U]$
and apply a binary search on them as above to find an interval $(\rho_L,\rho_U]$
such that $\rho_L \leq \radrestricted{i}{j} \leq \rho_U$.
Then for any $r\in (\rho_L,\rho_U]$, $M_t(r)$ for each
  $t \in \{1, 2\}$ remains the same.

The total amount of time spent for this procedure can be bounded by
$O((m+n)\log(m+n) + T_d \log(m+n)) = O((m+n) \log^3(m+n) )$,
by Theorem~\ref{thm:decision_algorithm},
where $T_d$ denotes the time spent by each call of the decision algorithm.
\end{proof}

We proceed with the interval $(\rho_L,\rho_U]$ described in
Lemma~\ref{lem:narrowed_interval}.  Since $M_1(r)$ and $M_2(r)$ remain
the same for any $r\in (\rho_L,\rho_U]$, we write $M_1 = M_1(r)$ and
$M_2 = M_2(r)$.  The sets $M_1$ and $M_2$ can be computed by
Lemma~\ref{lem:time_decision_marking}.  Note that
$M_1 = M_1(\radrestricted{i}{j})$ and
$M_2 = M_2(\radrestricted{i}{j})$ since
$\radrestricted{i}{j} \in (\rho_L,\rho_U]$.  We then pick a reference
point $o_t(r)$ on $\bd \dintersection_t(r)$ as done in
Section~\ref{sec:decision} such that the trace of $o_t(r)$ over
$r\in (\rho_L,\rho_U]$ is a simple curve.  This is always possible
because the combinatorial structure of $\bd \dintersection_t(r)$ is
constant.  (See also the proof of Lemma~\ref{lem:reference}.)  Such a
choice of references $o_t(r)$ ensures that the order on the events in
$M_t$ remains the same as $r$ continuously increases unless the
positions of two distinct events in $M_t(r)$ coincide.

We are now interested in the order $\prec^*_t$ on the events in $M_t$
at $r = \radrestricted{i}{j}$.  In the following, we obtain a sorted
list of events in $M_t$ with respect to $\prec^*_t$ without knowing
the exact value of $\radrestricted{i}{j}$.

\paragraph{Deciding Whether or Not
  \texorpdfstring{$x \preceq^*_t x'$}{x<t x'} for
  \texorpdfstring{$x, x' \in M_t$}{x, x' in Mt}.}
Let $q = \defpoint{x}$ and $q' = \defpoint{x'}$.  The
order of $x$ and $x'$ over $r\in (\rho_L,\rho_U]$ may change only
when we have a nonempty intersection of
$\boundaryarc_t(r) \cap \bd D_r(q) \cap \bd D_r(q')$.  Let
$\rho_t(q, q')$ denote such a radius $r>0$ that
$\boundaryarc_t(r) \cap \bd D_r(q) \cap \bd D_r(q')$ is nonempty for
any two distinct $q, q'\in Q_I$.  Note that the intersection
$\boundaryarc_t(r) \cap \bd D_r(q) \cap \bd D_r(q')$ at
$r=\rho_t(q, q')$ forms a single point $c$ and $D_{\rho_t(q,q')}(c)$
is the smallest-radius geodesic disk containing $Q_t\cup \{q, q'\}$.
Thus, the value $\rho_t(q, q')$ is uniquely determined.

\begin{lemma}
  \label{lem:time_sorting}
  For any two distinct points $q_1, q_2\in Q_I$, we can decide whether
  or not $\rho_t(q_1, q_2) \in (\rho_L,\rho_U]$ in
  $O(\log (m+n)\log n)$ time after a linear-time processing of $P$ for
  $t\in\{1, 2\}$.  If $\rho_t(q_1, q_2) \in (\rho_L,\rho_U]$, the
  value of $\rho_t(q_1, q_2)$ can be computed in the same time bound.
\end{lemma}
\begin{proof}
Let $c$ be the unique point in $\boundaryarc_t(\rho) \cap \bd D_\rho(q_1) \cap \bd D_\rho(q_2)$
for $\rho=\rho_t(q_1,q_2)$. It lies on the bisector of $q_1$ and $q_2$.
Consider a part $b$ of the bisector of $q_1$ and $q_2$ lying between one endpoint and the midpoint of $\pi(q_1,q_2)$. We will show how to compute $c$ in the case of $c\in b$. If it is not the case, we choose
the other part of the bisector, apply this algorithm again, and compute $c$.
In the following, we compute $c$ assuming that $\rho \in (\rho_L,\rho_U]$.
Once we compute $c$, we test whether $c$ is in
$\boundaryarc_t(\rho) \cap \bd D_\rho(q_1) \cap \bd D_\rho(q_2)$.
If so, we conclude that $\rho$ is in $(\rho_L,\rho_U]$.
Otherwise, we conclude that $\rho$ is not in $(\rho_L,\rho_U]$.

We apply a binary search on the arcs of $b$ of $q_1$ and $q_2$.  We call
each endpoint of the arcs of $b$ a \emph{breakpoint} of $b$.  Given
$b$ and a breakpoint $x$ of $b$, we can decide if $c$ comes
before $x$ from one endpoint of $b$ in $O(\log (m+n))$ time
as follows. First, we compute $d(q_1,x)$ in $O(\log n)$ time using the
shortest-path data structure of linear size~\cite{shortest-path}.  We
also compute the smallest value $\rho'$ with $x\in D_{\rho'}(q_1)$. To
do this, we find the cell of $\fvd{Q_1}$ containing $x$ and compute
the geodesic distance between $x$ and the site of the cell in
$O(\log(m+n))$ time.  This distance is the smallest value $\rho'$ by
definition.  The point $c$ lies between $x$ and the midpoint of
$\pi(q_1,q_2)$ if and only if $d(q_1,x)$ is larger than
$\rho'$. Therefore, we can decide if $c$ comes before $x$
from one endpoint of $b$ in $O(\log (m+n))$ time, and we can find the
arc of $b$ containing $c$ in $O(\log (m+n)\log n)$ time.  Once we find
the arc of $b$ containing $c$, we can compute $c$ in constant time.

However, it already takes $\Omega(n)$ time in the worst case to compute $b$.
Instead of computing $b$ explicitly, we present a procedure to compute an approximate
median of the breakpoints of a part of $b$ in $O(\log n)$ time.  More
precisely, given two breakpoints $x_1$ and $x_2$ of $b$, we compute a
breakpoint $x$ of $b$ such that the number of the breakpoints lying
between $x_1$ and $x$ is a constant fraction of the number of the
breakpoints lying between $x_1$ and $x_2$. To do this, we use the data
structure by Guibas and
Hershberger~\cite{shortest-path,Hersh-shortest-1991} constructed on
$P$.  This data structure has linear size and supports an
$O(\log n)$-time shortest path query between a source and a
destination.  In the preprocessing, they precompute a number of shortest
paths such that for any two points $p$ and $q$ in $P$, the shortest
path $\pi(p,q)$ consists of $O(\log n)$ subchains of the precomputed
shortest paths and $O(\log n)$ additional edges in linear time.  In
the query algorithm, the structure finds such subchains and edges connecting
them in $O(\log n)$ time.  Finally, it returns the
shortest path between two query points represented as a binary tree of
height $O(\log n)$~\cite{Hersh-shortest-1991}.  Therefore, we can
apply a binary search on the vertices of the shortest path between any
two points as follows.

We compute the midpoint of $\pi(q_1,q_2)$ by applying a binary search on
the edges of $\pi(q_1,q_2)$, and compute the endpoints of the
bisecting curve of $\pi(q_1,q_2)$ by applying a binary search on the
edges of $P$ in $O(\log^2 n)$ time. Now we have the endpoints $x_1$
and $x_2$ of $b$.  Using the data structure by Guibas and Hershberger,
we can compute two points $q_1'$ and $q_2'$ in $O(\log n)$ time such
that $\pi(q_t,q_t')$ is the maximal common path of the shortest paths
between $q_t$ and the endpoints of $b$ for $t=1,2$.

A breakpoint $x$ of $b$ corresponds to an edge $vv'$ of
$\pi(q_t',x_1)\cup\pi(q_t',x_2)$ for $t=1,2$ such that $\pi(q_t,x)$
contains the line segments containing $v$, $v'$ and $x$.  Moreover,
each edge of $\pi(q_t',x_s)$ induces a breakpoint of $b$, and the
breakpoints induced by such edges appear on $b$ in order for $t=1,2$
and $s=1,2$. Therefore, we can apply a binary search on the edges of
$\pi(q_t',x_s)$. We choose the median of the edges of $\pi(q_t',x_s)$
and compute the breakpoint induced by the median edge. We do
this for the four paths $\pi(q_t',x_s)$. Then one of them is an
approximate median of the breakpoints of the part of $b$ lying between
$x_1$ and $x_2$. Therefore, we can compute an approximate median in
$O(\log n)$ time, and determine which side of the approximate median
contains $c$ in $O(\log n)$ time. In total, we can find $c$ in
$O(\log (m+n)\log n)$ time.
\end{proof}

If $\rho_t(q, q') \notin (\rho_L,\rho_U]$, then the order of $x$ and $x'$ can be determined
by computing their positions at $r= \rho_L$ or $\rho_U$.
Otherwise,
we can decide whether or not $x \preceq^*_t x'$ by running the decision algorithm
for input $(i, j, \rho_t(q, q'))$, once we know the value $\rho_t(q, q')$.

\paragraph{Sorting the Events in \texorpdfstring{$M_t$}{Mt} with Respect to \texorpdfstring{$\prec^*_t$}{<t}.}
This can be done in $O(T_c\cdot m \log m)$ time, where $T_c$ denotes the time required
to compare two events as above.
We present a more efficient method that applies a parallel sorting algorithm
  due to Cole~\cite{parallel-sorting}.
Cole gave a parallel algorithm for sorting $N$ elements in
$O(\log N)$ time using $O(N)$ processors.
In Cole's algorithm, we need to apply $O(m)$ comparisons at each
iteration, while comparisons in each iteration are independent of one another.
For each iteration, we compute the values of $\rho_t(\defpoint{x}, \defpoint{x'})$
that are necessary for the $O(m)$ comparisons of $x, x'\in M_t$, 
and sort them in increasing order.
On the sorted list of the values, we apply a binary search using
the decision algorithm.
Then we complete the comparisons in each
iteration in time $O(m \log(m+n)\log n + T_d \log m)$ by Lemma~\ref{lem:time_sorting},
where $T_d$ denotes the time taken by the decision algorithm.
Since Cole's algorithm requires $O(\log m)$ iterations in total,
the total running time for sorting the events in $M_t$ is
$O(m\log (m+n) \log m \log n + T_d \log^2 m)$.

\paragraph{Computing \texorpdfstring{$\radrestricted{i}{j}$}{r*ij} and a Corresponding Two-Center.}
For any two neighboring events $x$ and $x'$ in $M_t$ with respect to $\prec^*_t$,
we call the value of $\rho_t(\defpoint{x}, \defpoint{x'})$ a \emph{critical radius}
if it belongs to $(\rho_L,\rho_U]$.
Let $R$ be the set of all critical radii, including $\rho_L$ and $\rho_U$.
\begin{lemma} \label{lem:critical_radii}
$R$ contains the radius $\radrestricted{i}{j}$ of an optimal $(i,j)$-restricted two-center.
\end{lemma}
\begin{proof}
  Assume to the contrary that $\radrestricted{i}{j} \notin R$.  This
  implies no coincidence between the positions of the events in
  $M_1(r)$ and $M_2(r)$, respectively, at $r=\radrestricted{i}{j}$.
  That is, the intersection points
  $\bd D_{\radrestricted{i}{j}}(q) \cap
  \boundaryarc_t(\radrestricted{i}{j})$ are all distinct for all
  $q\in Q_I$.  Let $(c_1, c_2)$ be an optimal $(i, j)$-restricted
  two-center of $Q$.  And let
  $Q'_t := D_{\radrestricted{i}{j}}(c_t) \cap Q$ and
  $\mathcal{I}'_t := \bigcap_{q\in Q'_t} D_{\radrestricted{i}{j}}(q)$
  for each $t\in \{1, 2\}$.  We observe that the intersection
  $\mathcal{I}'_t$ of geodesic disks is nonempty
  since the positions of the events in $M_t(\radrestricted{i}{j})$ are
  all distinct on $\boundaryarc_t(\radrestricted{i}{j})$.  This
  implies that there exists a sufficiently small positive
  $\epsilon > 0$ such that
  $\bigcap_{q\in Q'_t} D_{\radrestricted{i}{j}-\epsilon}(q)$ is still
  nonempty.  We pick any point $c'_t$ from the intersection of
  shrunken disks.  Then, $D_{\radrestricted{i}{j}-\epsilon}(c'_t)$
  contains all points in $Q'_t$ for each $t\in\{1, 2\}$.  This
  contradicts to the minimality of $\radrestricted{i}{j}$.  Thus, the
  optimal radius $\radrestricted{i}{j}$ must be contained in $R$.
\end{proof}

Hence, $\radrestricted{i}{j}$ is exactly the smallest value $\rho \in R$
such that there exists an $(i, j, \rho)$-restricted two-center.
The last step of our optimization algorithm thus performs a binary search on $R$
using the decision algorithm.

This completes the description of the optimization algorithm and we conclude the following.
\begin{theorem} \label{thm:optimization_algorithm} Given a partition
  pair $(i, j)$ and an assistant interval $(r_L,r_U]$, an optimal
  $(i,j)$-restricted two-center of $Q$ can be computed in
  $O((m+n)\log^3(m+n)\log m)$ time using $O(m+n)$ space, provided that
  $r_L \leq \radrestricted{i}{j} \leq r_U$.
\end{theorem}
\begin{proof}
The correctness follows from the arguments we have discussed above.
Thus, it computes an optimal $(i, j)$-restricted two-center of $Q$ correctly
if $\radrestricted{i}{j} \in (r_L,r_U]$; otherwise, it reports that $\radrestricted{i}{j} > r_U$.

The time and space complexities of our optimization algorithm are bounded as follows.
First, in Lemma~\ref{lem:critical_radii}, we spend $O((m+n) \log^3(m+n))$ time and $O(m+n)$ space
to compute $(\rho_L,\rho_U]$.
Second, we sort $M_t$ in $O(m\log (m+n) \log m \log n + T_d \log^2 m)$ time,
where $T_d$ denotes the time spent by the decision algorithm.
Finally, we compute the set $R$ of critical radii in $O(m \log(m+n) \log n)$ time
and do a binary search on $R$ in $O(T_d \cdot \log m)$ time, since $|R| = O(m)$.
We have $T_d = O((m+n) \log^2 (m+n))$ using $O(m+n)$ space by Theorem~\ref{thm:decision_algorithm}.
Thus, the total time complexity is bounded by $O((m+n)\log^3(m+n)\log m )$,
and the space complexity is bounded by $O(m+n)$.
\end{proof}

\section{Computing an Optimal Two-Center of Points}
\label{sec:algorithm}

Finally, we present an algorithm that computes an optimal two-center of $Q$ with
respect to $P$. As the optimization algorithm described in Section~\ref{sec:opt}
works with a fixed partition pair, we can find an optimal two-center by
trying all partition pairs $(i, j)$ 
once an assistant interval $(r_L,r_U]$ is computed.
In the following, we show how to choose $O(m)$ partition pairs one of which is
an optimal pair.

\subsection{Finding Candidate Pairs}
\label{sec:candidate}
In this subsection, we choose $O(m)$ pairs of indices, which we will call candidate partition pairs,
such that one of them is an optimal partition pair.
We will see that these pairs are obtained in a way similar to
an algorithm described~\cite{2-centersimple}.
Oh et al.~\cite{2-centersimple} presented an algorithm for computing
a geodesic two-center of a simple $n$-gon $P$.
This problem is equivalent to computing two  geodesic disks $D_1$ and $D_2$
of the minimum radius
whose union covers the whole polygon $P$.
For the purpose, they compute a set of $O(n)$ pairs $(e,e')$ of edges of $P$
such that both $e$ and $e'$ intersect the intersection of
$D_1$ and $D_2$.
The only difference in this paper is that we consider the extreme vertices of $Q$ while the algorithm in~\cite{2-centersimple} considers the vertices of the input polygon.

We define \emph{candidate (partition) pairs} as follows.
Recall that the extreme points of $Q$ are sorted in clockwise order along $\CHPQ$.
We denote the sequence by $\langle v_1,\ldots,v_k\rangle$ with $1<k\leq m$.
We use an index $i$ and its corresponding vertex $v_i$ interchangeably.
For instance, we sometimes say that an index comes before another index along $\bd \CHPQ$ in clockwise order from an index.
For two indices $i$ and $j$, we use $\radbd{i}{j}$
to denote the radius of $\subpolygon{i}{j}$.
Let $f(\cdot)$ be the function which maps each index $i$ with $1\leq i\leq k$ to
the set of indices $j$ that minimize $\maxrad{i}{j}$.
It is possible that there is more than one index $j$ that minimizes
$\maxrad{i}{j}$. Moreover, such vertices
appear on the boundary of $P$ consecutively.

We use $f_{cw}(i)$ to denote the set of all indices that come
after $i$ and before any index in $f(i)$ in clockwise order.
Similarly, we use $f_{ccw}(i)$ to denote the set of all indices that
come after $i$ and before any vertex in $f(i)$ in counterclockwise order.
The three sets $f_{ccw}(i)$, $f(i)$ and $f_{cw}(i)$ are pairwise disjoint
by the fact that $i \notin f(i)$ and
by the monotonicity of $\radbd{i}{\cdot}$ and $\radbd{\cdot}{i}$.

For an index $k$,
let $\vcw{k}$ be the last index of $f_{cw}(v_k)$ from $v_k$ in clockwise order
and $\vccw{k}$ be the first index of $f_{ccw}(v_{k})$ from $v_{k}$
in clockwise order.
Given an index $i$, an index $j$ is called a
\emph{candidate index of $i$} if it belongs to one of the following
two types:
\begin{enumerate}[(1)]
	\item $\{j, j+1\} \cap \{ \vccw{i}, \vcw{i+1}\} \neq \emptyset$.
	\item
	Both $v_j$ and $v_{j+1}$ lie on the chain
	$\subchain{\vccw{i}}{\vcw{i+1}}$
	with $\{j, j+1\} \cap \{ \vccw{i}, \vcw{i+1}\} = \emptyset$,
	and $\vccw{i}$ comes before $\vcw{i+1}$ from $i$ in clockwise order.
\end{enumerate}
A pair $(i, j)$ of indices is called a \emph{candidate partition pair} if
$j$ is a candidate index of $i$.

\begin{lemma}
	\label{lem:candidate}
	The number of candidate pairs is $O(m)$, and all candidate pairs can be computed in $O(m(m+n)\log (m+n))$ time using $O(m+n)$ space.
\end{lemma}
\begin{proof}
Since $\vccw{i}$ and $\vcw{i}$ are uniquely defined
for each index $1 \leq i \leq k$,
the total number of candidate pairs of type (1) is at most $4k \leq 4m$.

Now we consider the candidate pairs of type (2).
Assume that for an index $j$ with $1 \leq j \leq k$
there are two distinct indices $i$ and $i'$
such that $j$ is a candidate index of type (2) of both $i$ and $i'$,
that is, both $(i, j)$ and $(i', j)$ are candidate pairs of type (2).
Without loss of generality, we assume that $v_i$ comes before $v_{i'}$
in clockwise order from $v_j$.
Since they are of type (2),
both $v_j$ and $v_{j+1}$ are contained in the intersection of
$\subchain{\vccw{i}}{\vcw{i+1}}$
and $\subchain{\vccw{i'}}{\vcw{i'+1}}$.

We argue that $v_j$ lies on $\subchain{v_{i+1}}{\vcw{i+1}}$.
Suppose that $v_j \in \subchain{\vcw{i+1}}{v_{i+1}} \setminus \{ \vcw{i+1}, v_{i+1}\}$, for the sake of a contradiction.
Then, since $e_k$ is contained in $\subchain{\vccw{i}}{\vcw{i+1}}$,
the vertex $v_j$ lies in the interior of $\subchain{v_{k+1}}{v_{i+1}}$.
This contradicts the fact that $e_i$ comes before $e_j$ in clockwise order
from $e_k$.
Therefore, $v_{i'}$ lies on $\subchain{v_{i+1}}{\vcw{i+1}}$,
which implies that $v_{i'} \in f_{cw}(v_{i+1})$.
Consequently,
$\vccw{i'}$ lies in $\subchain{\vcw{i+1}}{v_{i'}}$.
Since both $v_j$ and $v_{j+1}$ are contained in
$\subchain{\vccw{i}}{\vcw{i+1}}$,
$\vcw{i+1}$ lies in
$\subchain{v_{j+1}}{v_i}$.
This implies that either $v_j$ or $v_{j+1}$ is not contained in
$\subchain{\vccw{i'}}{\vcw{i'+1}}$, which is a contradiction.
Therefore, for each index $1 \leq j \leq k$,
there is at most one index $i$ such that
$j$ is a candidate index of $i$ that is of type (2).
This implies that there are at most $m$ candidate pairs of type (2).

Now, we present an algorithm that computes all candidate pairs
in $O(m(m+n)\log (m+n))$ time using $O(m+n)$ space.
First, $\vccw{i}$ and $\vcw{i}$ can be computed in
$O((m+n)\log (m+n))$ time for each index $1\leq i \leq k$:
this can be done
by a binary search based on the monotonicity of $\radbd{i}{j}$
using a linear-time algorithm computing the radius
of a simple polygon~\cite{1-center}.
Thus, it takes $O(k(m+n)\log (m+m)) = O(m(m+n)\log(m+m))$ time
for computing  $\vccw{i}$ and $\vcw{i}$ for all indices $1\leq i\leq k$.

Next, we compute the set of all candidate pairs.
For each index $i$, we collect all indices $j$ such that
both $v_j$ and $v_{j+1}$ lie on the chain
$\subchain{\vccw{i}}{\vcw{i+1}}$
with $\{j, j+1\} \cap \{ \vccw{i}, \vcw{i+1}\} = \emptyset$
if $\vccw{i}$ comes before $\vcw{i+1}$ from $i$ in clockwise order.
Otherwise, we collect the four indices:
$\vccw{i}$, $\vcw{i+1}$, $\vccw{i} - 1$, and $\vcw{i+1} - 1$.
The indices collected in the former case correspond to
candidate indices of $i$ that are of type (2),
while those in the latter case correspond to candidate indices of $i$
that are of type (1).
This takes only time linear to the number of candidate pairs.

Hence, the total time spent for computing all candidate pairs
is $O(m(m+n) \log (m+n))$.
The space complexity is bounded by $O(m+n)$.
\end{proof}

The following is the key observation on candidate pairs.
\begin{lemma}
	\label{lem:candidate_pairs}
	There exists an optimal partition pair that is a candidate pair.
\end{lemma}
\begin{proof}
	Let $(c_1, c_2)$ be an optimal two-center of $Q$ with respect to $P$
and let $r^*$ be the smallest radius such that $Q$ is contained in $D_{r^*}(c_1) \cup D_{r^*}(c_2)$.
Let $(i,j)$ be the optimal partition pair of $D_{r^*}(c_1)$ and $D_{r^*}(c_2)$.
There is a pair $(\alpha, \beta)$ of points with $\alpha \in
\pi(v_{i },v_{i +1})$ and $\beta \in \pi(v_{j },v_{j +1})$ such that
the radii of $Q \cap \subpolygon{\beta}{\alpha}$ and $Q\cap\subpolygon{\alpha}{\beta}$, say
$r_1$ and $r_2$, are at most $r^*$.

Without loss of generality, we assume that $r_1\leq r_2=r^*$.
Consider every pair $(\alpha',\beta')$ of points in $\subchain{\alpha}{\beta}$ such that
the radius of $Q\cap \subpolygon{\beta'}{\alpha'}$ is at most the radius of $Q\cap \subpolygon{\alpha'}{\beta'}$,
and $\alpha$, $\alpha'$, $\beta'$ and $\beta$ lie in clockwise order along $\bd \CHPQ$.
Among them, we choose the one $(\alpha',\beta')$ that minimizes the length of $\subchain{\alpha'}{\beta'}$,
and redefine $(\alpha,\beta)$ to be this pair.
Then we redefine $i, j$ and $r_1, r_2$ accordingly. By construction, we still have $r_1\leq r_2=r^*$.
We claim the followings:
\begin{enumerate}[1.] \denseitems
	\item $\vccw{i}\in\subchain{i+1}{j+1}$, and
	\item $\vcw{i+1}\in\subchain{j}{i}$.
\end{enumerate}

Recall that $\radbd{i}{j}$ is the radius of $\subpolygon{i}{j}$.
Claim 1 holds because
\[\radbd{j +1}{i}\leq r_1\leq r_2\leq \radbd{i}{j+1}.\]
The first inequality holds because $\subpolygon{j+1}{i}$ is contained in the minimum-radius geodesic disk
containing $Q\cap \subpolygon{\beta}{\alpha}$.
The last inequality holds because $Q\cap\subpolygon{\alpha}{\beta}$ is contained in $\subpolygon{i}{j+1}$.

For Claim 2, we observe that $\radbd{i+1}{j}$ is at most $\radbd{j}{i+1}$.
Assume to the contrary that $\radbd{i+1}{j}>\radbd{j}{i+1}$.
Then we have
\[r_1\leq\radbd{\beta}{\alpha}\leq
\radbd{j}{i+1}< \radbd{i+1}{j}\leq r_2.\]
The first inequality holds simply because $Q\cap\subpolygon{\beta}{\alpha}$ is contained in  $\subpolygon{\beta}{\alpha}$.
The second inequality holds because  $Q\cap\subpolygon{\beta}{\alpha}$ is contained in $\subpolygon{j}{i+1}$.
The third inequality holds by the assumption.
The last inequality holds $\subpolygon{i+1}{j}$ is contained in the minimum-radius enclosing geodesic disk containing $Q\cap \subpolygon{\alpha}{\beta}$. Thus, we have $\radbd{j}{i+1}\leq r_2$.
Let $\tilde{\alpha}$ be $v_{i+1}$ and $\tilde{\beta}$ be a point on  $\pi(v_{j-1},v_{j})$ and sufficiently close to $v_{j}$.
The radius of $Q\cap \subpolygon{\tilde{\alpha}}{\tilde{\beta}}$ is at most $c_2$ by definition,
and the radius of $Q\cap \subpolygon{\tilde{\beta}}{\tilde{\alpha}}$ is at most $\radbd{j}{i+1}$, which is less than $c_2$.
This contradicts the construction of $(\alpha,\beta)$.
Therefore,  $\radbd{i+1}{j}$ is at most $\radbd{j}{i+1}$, and Claim 2 holds.

Due to the two claims, $\vcw{i+1}$ appears after $\vccw{i}$ as
we move clockwise from $v_{i}$ along $\bd \CHPQ$ unless $(i,j)$ is a candidate partition pair of type (1).
Moreover, both $v_{j}$ or $v_{j+1}$ lie in the interior of $\subchain{\vccw{i}}{\vcw{i+1}}$.
Therefore, $(i, j)$ is a candidate partition pair of type (2) unless it is a candidate partition pair of type (1).
\end{proof}

\subsection{Applying the Optimization Algorithm for Candidate Pairs}
To compute the optimal radius $r^*$, we apply the optimization
algorithm in Section~\ref{sec:opt} on the set of all candidate pairs.
Let $C$ be the set of candidate pairs.
By Lemma~\ref{lem:candidate_pairs}, we have
$r^* = \min_{(i,j)\in C}{\radrestricted{i}{j}}$.
To apply the optimization algorithm, we have to compute an assistant interval $(r_L,r_U]$
satisfying that $r^* \in (r_L,r_U]$ and the combinatorial structure of $\bd D_r(q)$ for each $q\in Q$
remains the same for all $r\in (r_L,r_U]$.
\begin{lemma} \label{lem:assistant_interval}
 An assistant interval $(r_L,r_U]$ such that the combinatorial structure of
 $\bd D_r(q)$ for each point $q\in Q$ remains the same can be computed in
 $O(m(m+n)\log^3(m+n))$ time using $O(m+n)$ space.
\end{lemma}
\begin{proof}
Let $q \in Q$ and $r \geq 0$.
Recall that the boundary $\bd D_r(q)$ of geodesic disk $D_r(q)$ consists of
boundary arcs and line segments that are portions of $\bd P$.
A boundary arc is a circular arc with endpoints lying
on an edge of $P$, of $\fvd{Q_1}$, or of $\fvd{Q_2}$,
while each non-extreme part is a portion of an edge of $P$.
Hence, the combinatorial structure of $\bd D_r(q)$ can be represented by
a (cyclic) sequence of edges of $P$, $\fvd{Q_1}$, and $\fvd{Q_2}$ on which each endpoint of
the boundary arcs or the line segments in $\bd D_r(q)$ lies.

Imagine we blow the geodesic disk $D_r(q)$ by increasing $r$ from $0$ continuously.
The combinatorial structure of $D_r(q)$ changes
exactly when a new boundary arc or non-extreme segment appears
or an existing arc or segment disappears.
We observe that such a change occurs exactly when
when the disk $D_r(q)$ touches a vertex $v$ of $P$
or an edge of $P$.
Both cases can be captured by the shortest path map $\spm{q}$ of $q$:
when $D_r(q)$ touches a vertex $x$ of $\spm{q}$,
or equivalently, when $r = d(q, x)$ for some vertex $x$ of $\spm{q}$.
This implies that for an assistant interval $(r_L,r_U]$ with the desired property,
there is no $q \in Q$ and no vertex $x$ of $\spm{q}$ such that
$d(q, x) \in (r_L,r_U]$.

In order to compute such an interval,
we initially set $r_L$ to $0$ and $r_U$ to $\infty$,
and perform a binary search on the set of values $d(q, x)$ for  all $q\in Q$
and all vertices $x$ of $\spm{q}$ as follows:
\begin{enumerate} [(i)] \denseitems
\item For each $q\in Q$, build the shortest path map $\spm{q}$,
 compute $d(q, x)$ for all vertices $x$ of $\spm{q}$,
 and find the median $d_q$ of those values that lie in $(r_L,r_U]$, if any.
\item Find the median $r$ of the medians $d_q$ for all $q\in Q$.
\item Check if $r \geq r^*$ as follows:
  Execute our decision algorithm for $(i, j, r)$ described in Section~\ref{sec:decision}
     for all candidate pairs $(i, j)$ computed by Lemma~\ref{lem:candidate}.
  By Lemma~\ref{lem:candidate_pairs},
  if there exists a candidate pair $(i, j)$ such that the decision algorithm returns ``yes'',
  then we have $r \geq r^*$; otherwise, we have $r < r^*$.
\item  If $r \geq r^*$, then set $r_U$ to $r$; otherwise, set $r_L$ to $r$.
\item Repeat from Step (i) until the interval $(r_L,r_U]$ does not change any longer.
\end{enumerate}
The above procedure iteratively reduces the interval $(r_L,r_U]$ while
keeping the property that $r^* \in (r_L,r_U]$.
In addition, $(r_L,r_U]$ eventually contains no distance $d(q, x)$
for any $q\in Q$ and a vertex $x$ of $\spm{x}$.
Hence, the procedure guarantees to reduce
the interval $(r_L,r_U]$ with the desired property.

Now, we analyze its time and space complexity.
As a preprocessing, compute the set of all candidate pairs
in $O(m(m+n)\log (m+n))$ time and $O(m+n)$ space by Lemma~\ref{lem:candidate}.
Note that there are $O(m)$ candidate pairs.
We spend $O(mn \log n)$ time for Steps (i) and (ii),
and $O(m(m+n)\log^2 (m+n))$ time for Step (iii) by Theorem~\ref{thm:decision_algorithm}.
By selecting the median $r$ of medians, the procedure terminates after $O(\log (n+m))$ iterations.
Therefore, it takes $O(m(m+n)\log^3(m+n))$ time in total.
Also, it is not difficult to see that this procedure can be implemented using $O(m+n)$ space
if we do not keep $\spm{q}$ and the distances $d(q,x)$.
\end{proof}

Now, we are ready to execute our optimization algorithm.
We run the optimization algorithm for each $(i, j) \in C$ and
find the minimum of $\radrestricted{i}{j}$ over $(i, j)\in C$.

\begin{theorem}
  \label{thm:main}
  An optimal two-center of $m$ points with respect to a simple $n$-gon
  can be computed in $O(m(m+n)\log^3(m+n) \log m)$ time using $O(m+n)$
  space.
\end{theorem}
\begin{proof}
  The correctness of our algorithm follows from the arguments above.
  We thus focus on analyzing the time complexity.

  In our algorithm, we first compute the set of candidate pairs in
  $O(m(m+n)\log (m+n))$ time by Lemma~\ref{lem:candidate}.  Before
  running the optimization algorithm, we spend $O(m(m+n)\log^3(m+n))$
  time by Lemma~\ref{lem:assistant_interval} to compute an assistant
  interval $(r_L,r_U]$.  The main procedure consists of $O(m)$ calls
  of the optimization algorithm, which takes
  $O(m(m+n)\log^3(m+n)\log m)$ time by
  Theorem~\ref{thm:optimization_algorithm}.  Also, the space usage is
  bounded by $O(m+n)$ in every step.
\end{proof}



{
\bibliographystyle{abbrv}
\bibliography{paper}
}

\end{document}